%% file: slsh.tex
\newif\ifsubmission
\submissionfalse

\newif\ifanonymous
\anonymoustrue

\documentclass[sigconf]{acmart}
\renewcommand\footnotetextcopyrightpermission[1]{} 
\usepackage{booktabs}       

\usepackage{graphicx}
\usepackage{amsmath}
\usepackage{multirow}

\usepackage{amsfonts}
\usepackage{amssymb}
\usepackage{algorithmic}
\usepackage[ruled,vlined]{algorithm2e}
\usepackage{tikz}
\usepackage{cancel}
\usepackage{todonotes}
\usepackage{tipa}
\usepackage{color}
\usepackage{subcaption}

\usepackage{enumitem}
\usepackage{listings}
\usepackage{footnote}
\usepackage{numprint}
\usepackage{wrapfig}
\usepackage{caption}
\usepackage{subcaption}
\usepackage{balance}
\usetikzlibrary{matrix,shapes,arrows,positioning,chains, calc}

\graphicspath{{figs/}}
\input{macros}

\newcommand\blfootnote[1]{%
  \begingroup
  \renewcommand\thefootnote{}\footnote{#1}%
  \addtocounter{footnote}{-1}%
  \endgroup
}

\setlist[itemize]{noitemsep,nolistsep}

\newcounter{counter}[subsubsection]

\makeatletter
\def\BState{\State\hskip-\ALG@thistlm}
\makeatother

\newtheorem{thm}{Theorem}[section]

\newtheorem{prop}[thm]{Proposition}
\newtheorem{cor}{Corollary}

\newtheorem{defn}{definition}[section]
\begin{document}

\title{Sub-Linear Privacy-Preserving Near-Neighbor Search}

\author{M. Sadegh Riazi{\Large $^*$}}
\affiliation{UC San Diego}
\email{mriazi@ucsd.edu}

\author{Beidi Chen{\Large $^*$}}
\affiliation{Rice University}
\email{beidi.chen@rice.edu}

\author{Anshumali Shrivastava}
\affiliation{Rice University}
\email{anshumali@rice.edu}

\author{Dan Wallach}
\affiliation{Rice University}
\email{dwallach@cs.rice.edu}

\author{Farinaz Koushanfar}
\affiliation{UC San Diego}
\email{farinaz@ucsd.edu}


\begin{abstract}
In Near-Neighbor Search (NNS), a client queries a database (held by a server) for the most similar data (near-neighbors) given a certain similarity metric. The Privacy-Preserving variant (PP-NNS) requires that neither server nor the client shall learn information about the other party's data except what can be inferred from the outcome of NNS.
The overwhelming growth in the size of current datasets and the lack of a truly secure server in the online world render the existing solutions impractical; either due to their high computational requirements or non-realistic assumptions which potentially compromise privacy. PP-NNS having query time {\it sub-linear} in the size of the database has been suggested as an open research direction by Li et al. (CCSW'15). In this paper, we provide the first such algorithm, called Privacy-Preserving Locality Sensitive Indexing (\sys{}) which has a sub-linear query time and the ability to handle honest-but-curious parties.
At the heart of our proposal lies a secure binary embedding scheme generated from a novel probabilistic transformation over locality sensitive hashing family.
We provide information theoretic bound for the privacy guarantees and support our theoretical claims using substantial empirical evidence on real-world datasets.
\end{abstract}

\maketitle


\section{Introduction}\label{sec:intro}
\blfootnote{{\Large $^*$} Equal contribution of the first two authors.}
Near-Neighbor Search (NNS) is one of the most fundamental and frequent tasks in large-scale data processing systems. In NNS problem, a server holds a collection of users' data; a new user's objective is to find all similar data to her query given a certain similarity metric.
NNS is used in personal recommendations of friends, events, movies, etc.~\cite{su2009survey}, online classification based on $K$-NN search, face recognition~\cite{sadeghi2009efficient}, secure biometric authentication~\cite{barni2010privacy,blanton2011secure}, privacy-preserving speech recognition~\cite{pathak2012privacy}.
The demand for privacy in big-data systems has led to an increasing interest in the problem of Privacy-Preserving Near-Neighbor Search (PP-NNS).
In PP-NNS, all of the clients' data must remain private to their respective owners. This implies that not only server(s), but also a new client who queries the database, should not learn information about other clients' data except the NNS result.

The above setting is natural and ubiquitous in the online world where matching and recommendations are common~\cite{polat2003privacy}. For example, on dating websites, a client is interested in finding similar profiles (near neighbors) without revealing her attributes to anyone.
Note that, it is problematic to assume any trusted server in real settings. 
A well-publicized recent example is Yahoo's massive leak, which compromised 500 million user accounts including private information such as phone number, date of birth, or even answers to security questions~\cite{yahoo}. It is, therefore, desirable that the protocol does not rely on the complete security of participating servers and even if data from the server is compromised, the user's information must remain secure.


Keeping in mind both big-data and modern security challenges, four main requirements have to be satisfied: (i) one shall not assume any trusted server, (ii) data owners (clients) are not trusted, (iii) modern datasets are very high dimensional, and (iv) the query time must be sub-linear (near constant) in the number of clients (or database size) in order to handle web-scale datasets. Sub-linear privacy-preserving solution without any trusted party is currently considered to be a critical, yet open, research direction as stated in a recent article~\cite{li2015exploring}.



Due to the importance of the PP-NNS problem, there have been many attempts to create a practical solution.
In theory, any function (e.g., NNS) with inputs from different parties can be evaluated securely without revealing the input of each party to another using Secure Function Evaluation (SFE) protocols such as Garbled Circuit (GC) protocol.
While the SFE protocols have been continuously improving in efficiency, they still suffer from huge execution times and massive communication between executive servers. In addition, realizing NNS with any of the SFE protocols faces the scalability issue. These protocols scale (at best) linearly with respect to the size of the database~\cite{rane2013privacy}, undermining requirement four.
As we describe later, we only utilize GC for a small part of the computation.

Supporting NNS on encrypted data is an active area of research~\cite{wong2009secure,hu2011processing,yao2013secure,elmehdwi2014secure}. Unfortunately, available crypto-based solutions fail to support high dimensional data and they usually require multiple rounds of communication between user and the server.
Mylar~\cite{popa2014building} is a system for web applications that works on top of encrypted data which is proved to be insecure by Grubbs et al.~\cite{grubbs2016breaking}.
One of the most adopted solutions is Asymmetric Scalar-Product-preserving Encryption (ASPE)~\cite{wong2009secure}. However, not only this scheme has linear query complexity in terms of the size of the database, it has been proven to be insecure against chosen plaintext attack by Yao et al.~\cite{yao2013secure}.
More generally, they have proved that secure NNS is at least as hard as Order Preserving Encryption (OPE). Since it has been proven that it is impossible to have secure OPE under standard security models~\cite{boldyreva2009order,boldyreva2011order}, {\it it is not feasible to have a Secure NNS under standard security models such as Ciphertext Indistinguishability under chosen Plaintext Attack (IND-CPA)}.
In this paper, we define $\epsilon$-security and show that our solution limits the information leakage (for any arbitrary upper bound) while having a practical sub-linear PP-NNS.

We propose Privacy-Preserving Locality Sensitive Indexing (\sys{}) as a practical solution for the sub-linear PP-NNS on high dimensional datasets.
Performing NNS on a very high dimensional database is a non-trivial task even when data privacy is not a constraint. For example, NNS algorithms based on k-d trees are marginally better than exhaustive search~\cite{indyk2004nearest} in high-dimensional data spaces.
Our solution has two main components: (i) a novel probabilistic transformation over locality sensitive hashing family (\sect{sec:secLSH}) and a (ii) secure black-box hash computation method based on the GC protocol (\sect{sec:secure_hash}).

Locality Sensitive Hashing (LSH) is the only line of work which guarantees sub-linear query time approximate near-neighbor search for high-dimensional datasets~\cite{Proc:Indyk_STOC98}.
One fundamental property of LSH-based binary embedding is that it preserves all pairwise distances with little distortion~\cite{johnson1984extensions}, eliminating the need for sharing original attributes. However, the bits of the binary embeddings have enough information to estimate any pairwise distance (or similarity) between any two users~\cite{Proc:Bayardo_WWW07}, which makes them unsuitable in settings with no trusted party. \emph{We argue that the ability to estimate all pairwise distances is sufficient but not necessary for the task of near-neighbor search.}
In fact, we show for the first time, that the ability to estimate distances compromises the security of LSH-based embedding; rendering them susceptible to ``triangulation'' attack (see Section~\ref{sec:traingulations}). In this work, we eliminate the vulnerability of LSH with minimal modification while not affecting the sub-linear property.


{\bf Contributions.} Our main contributions are as follows:
\begin{itemize}
	\item We propose the first algorithm for PP-NNS with query time sub-linear in the number of clients. No trusted party or server is needed for handling sensitive data.
	\item We introduce the first generic transformation which makes any given LSH scheme secure for public release in honest-but-curious adversary setting. 
	This advantage comes at no additional cost and we retain all the properties of LSH required for the sub-linear search.
	\item We give information theoretic guarantees on the security of the proposed approach. Our proposed transformation, analysis, and the information theoretic bounds are of independent theoretical interest.
	\item We provide a practical implementation of \emph{triangulation attack} for compromising the security of LSH signatures in high dimensions. Our attack is based on alternating projections. The proposed attack reveals the vulnerability and unnecessary information leakage by the LSH embeddings. In general, we experimentally verify that the ability to estimate all pairwise distances is sufficient for recovering original attributes.
	\item We support our theoretical claims using substantial empirical evidence on real-world datasets.
    We further provide the first thorough evaluation of accuracy-privacy trade-off and its comparison with noise-based privacy.
	Our scheme can process queries against a database of size 3 Billion entries in {\it real time} on a typical PC. Performing the same task with the state-of-the-art GC protocol requires an estimated time of $1.5\times 10^8$ seconds and $1.2\times 10^7$ GBytes of communication (see \sect{ssec:cwsfe}).
\end{itemize}

\section{Preliminaries and Background}\label{sec:LSH}
In this section, we briefly review our notation. Then, we discuss our threat model followed by a background on LSH. Finally, we explain how LSH is currently used for large-scale near-neighbor search when the server is trusted. Please refer to~\cite{Proc:Indyk_STOC98,indyk2006polylogarithmic} for more specific details.

\subsection{Key Notations and Terms}
A server holds a giant collection $\mathcal{C}$ of clients (or data owners), each represented by some $D$ dimensional attribute vectors, i.e., $\mathcal{C} \subset \mathbb{R}^D$. We are interested in finding the answers to queries. The objective is $$\arg\max_{x \in \mathcal{C}} Sim(x,q),$$ where $Sim(.,.)$ is a desired similarity measure. However, the process should prevent any given (possibly dishonest) client from inferring the attributes of other clients, except for the information that can be inferred from the answer of the NNS queries.

We interchangeably use the terms clients, users, data owners, vectors, and attributes. They all refer to the vectors in the collection $\mathcal{C}$. Unless otherwise stated, the hash functions $h$ will produce a $1$-bit output, i.e., $h(x) \in \{0,\ 1\}$. All the hash functions are probabilistic, and in particular, there is an underlying family (class) of hash functions $\mathcal{H}$ and $h$ is drawn uniformly from this family. The draw can be conveniently fixed using random seeds. Our protocol will require some $l$-bits embedding and each of these $l$-bits will be formed by concatenating $l$ independent draws $h_i \ i\in\{1,2,...,l\}$ from some family of hash functions. Similarity search and the near-neighbor search will mean the same thing. Similarity and distances can be converted into each other using the formula distance = 1 - similarity. For any hash function $h$, the event $h(x) = h(y)$, for given pair $x$ and $y$, will be referred to as the collision of hashes.

\subsection{Threat Model}
There are two types of parties involved in our model: servers and clients (data owners). The models in previous works, for example~\cite{li2015exploring}, usually consider trusted servers. 
In this paper, we assume Honest-but-Curious (HbC) adversary model for both data owners and servers. In this threat model, each part is assumed to follow the protocol but is curious to extract as much information as possible about other party's secret data. While we do not trust any server, we assume that the servers do not collude with each other. 
Please note that this is the exact security model of the prior art~\cite{elmehdwi2014secure}. We want to emphasize that the assumption of two non-colluding HbC servers is feasible since two servers can represent two different companies, e.g. Amazon and Microsoft. Due to the business reasons and the fact that any collusion will significantly damage their reputation, it would be very unlikely that two companies will collude since it would be against their interests.

The solutions based on the GC protocol, Fully Homomorphic Encryption (FHE), and Oblivious RAM (ORAM) do not leak {\it any} information about the database and the query~\cite{naveed2015inference} other than what can be inferred from the answer of NNS. All other solutions leak some information either in the setup phase (creating the database) or the query phase. Unfortunately, GC, FHE, and ORAM solutions are computationally too expensive to be employed in real-world~\cite{wong2009secure}.
In this paper we compare the performance of our proposed solution (\sys{}) with GC. In addition, we formalize $\epsilon$-security and prove that the information leakage in our scheme can be made as small as required by tuning a privacy parameter in the protocol. We also compare our solution to the noise addition-based techniques and illustrate, both experimentally and theoretically, that our solution has significantly higher precision/recall for the same security limits. Therefore, our work fills the gap between fast non-secure solutions and impractical but secure ones, providing a practical and controllable trade-off between efficiency and the privacy of users.

Note that the answer to NNS may reveal some information about the query and/or the database, regardless of implementation details and security guarantees of any protocol. For example, if client $i$ and $j$ are very close w.r.t similarity measure (near identical), then the near-neighbor query of client $i$ should return $j$ as the correct answer (with a high probability). A correct answer automatically reveals information that $j$'s attributes are likely to be very similar to $i$'s attributes (with a high probability) even without having knowledge of the other client's attributes. This kind of information leak cannot be avoided by any algorithm answering the near-neighbor query with a reasonable accuracy.

Privacy guarantees in the PP-NNS protocols all rely on the inherent assumption on bounded computations. Given unbounded computations, the adversary can enumerate the whole space of every possible vector and use near-neighbor query until the generated vector returns the target client as the neighbor. In high dimensions, this process will require exponential computations due to the curse of dimensionality, which turns out to be a boon for the privacy of NNS.

\subsection{Locality Sensitive Hashing}
\label{lsh}
A popular technique for approximate near-neighbor search uses the underlying theory of \emph{Locality Sensitive Hashing}~\cite{Proc:Indyk_STOC98}. LSH is a family of functions with the property that similar input objects in the domain of these functions have a higher probability of colliding in the range space than non-similar ones.
In formal terms, consider $\mathcal{H}$ a family of hash functions mapping $\mathbb{R}^{D}$ to some set $\mathcal{S}$.

\begin{defn} [\bf LSH Family]\ A family $\mathcal{H}$ is called\\
$(S_0,cS_0,p_1,p_2)$-sensitive if for any two points $x,y \in \mathbb{R}^{D}$ and $h$ chosen uniformly from $\mathcal{H}$ satisfies the following:

       \begin{itemize}
       	\item if $Sim(x,y)\ge S_0$ then ${Pr}(h(x) = h(y)) \ge p_1$
       	\item if $ Sim(x,y)\le cS_0$ then ${Pr}(h(x) = h(y)) \le p_2$
       \end{itemize}
\label{def:lsh}
\end{defn}

For approximate nearest neighbor search typically, $p_1 > p_2$ and $c < 1$ is needed. An LSH allows us to construct data structures that give provably efficient query time algorithms for the approximate near-neighbor problem with the associated similarity measure.

One sufficient condition for a hash family $\mathcal{H}$ to be a LSH family is that the \emph{\bf collision probability} ${Pr}_\mathcal{H}(h(x) = h(y))$ is monotonically increasing function of the similarity, i.e. \begin{equation}\label{eq:monotonic}{Pr}_\mathcal{H}(h(x) = h(y)) = f(Sim(x,y)),\end{equation} where f is a monotonically increasing function. In fact most of the popular known LSH families, such as MinHash (Section~\ref{sec:MinHash}) and SimHash (Section~\ref{sec:SimHash}), actually satisfy this stronger property. It can be noted that Equation~\ref{eq:monotonic} automatically guarantees the two required conditions in the Definition~\ref{def:lsh} for any $S_0$ and $c < 1$.

It was shown~\cite{Proc:Indyk_STOC98} that having an LSH family for a given similarity measure is sufficient for efficiently solving near-neighbor search in sub-linear time:
\begin{defn} Given a family of
	$(S_0, cS_0, p_1, p_2)$-sensitive hash functions, one can construct a data
	structure for c-NN with $O(n^\rho \log n)$ query time and space $O(n^{1+\rho})$, where $\rho = \frac{\log p_1}{ \log p_2}< 1$.
	\label{def:lshquery}
\end{defn}

\subsection{Popular LSH 1: Minwise Hashing (MinHash)}

\label{sec:MinHash}
One of the most popular measures of similarity between web documents is {\it resemblance} (or Jaccard similarity) $\mathcal{R}$~\cite{Proc:Broder}. This similarity measure is only defined over sets which can be equivalently thought of as binary vectors over the universe, with non-zeros indicating the existence of those elements

The resemblance similarity between two given sets $x$, $y \subseteq \Omega = \{1,2,...,|\Omega|\}$ is defined as

\begin{equation}
\mathcal{R} = \frac{|x \cap y|}{| x \cup y|} = \frac{a}{f_1+f_2 -a},
\end{equation}

\noindent where $f_1 = |x|$, $f_2 = |y|$, and $a = |x \cap y|$.

Minwise hashing~\cite{Proc:Broder_STOC98} is the LSH for resemblance similarity. The minwise hashing family applies a random permutation $\pi:\Omega \rightarrow \Omega$, on the given set $x$, and stores only the minimum value after the permutation mapping. Formally MinHash and its collision probability is given by
\begin{equation}\label{eq:MinHash}h_{\pi}^{min}(x) = \min(\pi(x)); \ \ Pr(h_{\pi}^{min}(x) = h_{\pi}^{min}(y))  = \mathcal{R}.
\end{equation}

%

\subsection{Popular LSH 2: Signed Random Projections (SimHash)}
\label{sec:SimHash}

SimHash is another popular LSH for the {\it cosine} similarity measure, which originates from the concept of Signed Random Projections (SRP)~\cite{Proc:Charikar,Book:Rajaraman_11, Proc:Henzinger_06}. Given a vector $x$, SRP utilizes a random $w$ vector with each component generated from i.i.d. normal distribution, i.e., $w_i \sim N(0,1)$, and only stores the sign of the projection. Formally,
\begin{equation}\label{eq:hsrp}h^{sign}_w(x) = sign(w^Tx).\end{equation}
It was shown in the seminal work~\cite{Article:Goemans_95} that collision under SRP satisfies the following equation:
\begin{equation}
\label{eq:srp}
Pr(h^{sign}_w(x) = h^{sign}_w(y)) = 1 - \frac{\theta}{\pi},
\end{equation}
where $\theta = cos^{-1}\left( \frac{x^Ty}{||x||_2\cdot ||y||_2}\right)$. The term $\frac{x^Ty}{||x||_2\cdot ||y||_2}$, is the cosine similarity.
There is a variant of SimHash that performs similar to the original one~\cite{Book:Rajaraman_11} where, instead of $w_i \sim N(0,1)$, we choose each $w_i$ independently as either +1 or -1 with probability $\frac{1}{2}$. 
Since $1 - \frac{\theta}{\pi}$ is monotonic with respect to cosine similarity $\mathcal{S}$, SimHash is a valid LSH.

\subsection{Mapping LSH to 1-bit}

LSH, such as MinHash, in general, generates an integer value, which is expensive from the storage perspective. It would gain a lot of benefits from having a single bit hashing schemes, or binary locality sensitive bits. It is also not difficult to obtain 1-bit LSH. The idea is to apply a random universal hash function to the LSH and map it to 1-bit.

A commonly used universal scheme is given by
\begin{equation}
h_{1bit}(x) = a\times x \ \ \text{mod} \ 2,
\end{equation}
where $a$ is an odd random number, see~\cite{Proc:Carter_STOC77} for more details. With this 1-bit mapping, any hashing output $h(x)$ can be converted to $1$-bit by applying universal 1-bit hash function $h_{1bit}(.)$. Collision probability of this new transformed 1-bit hashing scheme is given by
\begin{equation}\label{eq:1bitprob}
Pr(h_{1bit}(h(x)) = h_{1bit}(h(y)))= \frac{Pr(h(x)=h(y))+1}{2}.
\end{equation}

It is not difficult to show that $h_{1bit}(h(x))$ is also a valid LSH familiy for the same similarity measure associated with $h(.)$~\cite{Proc:Charikar,Proc:Shrivastava_ECML12}.
Another convenient (and efficient) $1$-bit rehashing is to use the parity, or the most significant bit, of $h_{\pi}^{min}(x)$ as 1-bit hash~\cite{Proc:Shrivastava_ECML12}.

\subsection{PP-NNS in Sub-linear Time with a Trusted Server}
\label{sec:LSHPRivacy}

In the trusted server settings, LSH-based protocols are well-known for sub-linear near-neighbor search~\cite{indyk2006polylogarithmic}. The protocol involves two major steps:
\begin{enumerate}
	\item {\bf Constructing Hash Tables (Pre-processing):} The trusted server fixes random seeds for hash functions (e.g., random permutation for MinHash). Every client $x \in \mathcal{C}$ sends its attributes to the server. The server computes the $l$-bit binary embedding $E(x)$, using appropriate (pre-chosen) LSH schemes $h_i(x)$s. Computing $l$ bits involves generating multiple $1$-bit hashes using independent randomization and concatenating them $E(x) = [h_1(x);h_2(x);...;h_l(x)]$, where $h_i(x)$ is an independent hashing scheme. The server also generates hash tables, as a part of preprocessing for sub-linear time search. New clients can be dynamically inserted into the tables.
	\item {\bf Sub-linear Search (Querying):} To find near-neighbors of any given query point $q$, the trusted server computes the $l$-bit embedding of $q$, $E(q)$. Due to the LSH property of $E$, it suffices to find points $y \in \mathcal{C}$ such that $E(q)$ and $E(y)$ are close in Hamming distance. Searching for close Hamming distance can be done very efficiently in sub-linear time by only probing few buckets in the pre-constructed hash tables~\cite{indyk2006polylogarithmic}.
\end{enumerate}

The above protocol requires a trusted server which handles all the data. The security relies on the fact that no client is allowed to see any part of the computation process. The sub-linearity of the search is due to the classical sub-linear LSH algorithm for Hamming distance search~\cite{Book:Rajaraman_11}.

\begin{figure*}[ht]
       \centering
       \includegraphics[width = 0.85\textwidth]{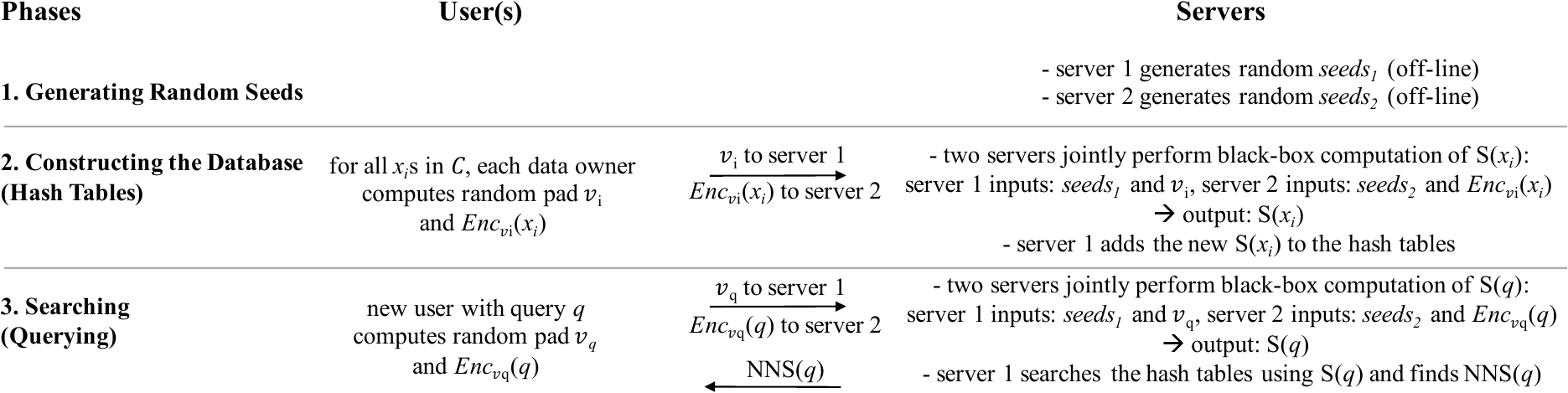}
       \caption{The \sys{} scheme consists of three phases: (i) generating random seeds, (ii) constructing the database, and (iii) searching phase.
       The internal mechanism of $S(.)$ is explained in~Seciton~\ref{sec:secLSH}.
       Black-box hash computation of $S(x)$ is described in Section~\ref{sec:secure_hash} which is based on the GC protocol and takes as input $Enc_v(x)$ (encryption of $x$) and the pad $v$.}
       \label{alg:HashPseudoAlgo}
\end{figure*}

\section{Challenges with Untrusted Server and Types of Attacks}\label{sec:challun}
For obtaining sub-linear solutions, we do not have many choices. LSH-based techniques are well-known methods that guarantee efficient sub-linear query time algorithms even in high dimensions~\cite{Gionis:1999}. Thus, one cannot hope to deviate from the philosophy of generating binary embeddings for data vectors that preserve original near-neighbors in the obtained Hamming space.

Conceptually, there can be three types of attacks when the server is untrusted: (i) querying the database and brute-forcing the space of inputs (ii) extracting the original attribute vector from the hashes using compressive sensing theory, and (iii) analyzing the combination of hashes and measuring their mutual correlation to estimate the original attribute of a user.

{\bf Brute-force/Probing Attacks.}
An attacker can ask for the hash embedding of a random attribute vector and check whether it is equal to another user's hash (if the database is compromised). However, exploring the entire input space is computationally infeasible. If each element of the attribute vector is represented as a 32-bit number, we have $(2^{32})^D$ possible unique inputs. National Institute of Standards and Technology (NIST) states that any attack that requires $2^{128}$ operations is computationally infeasible~\cite{nist}. For example, for the two datasets that are considered in this paper, $D\geqslant186$. Thus, there are $2^{5952}$ possible inputs which are far beyond the security standards. Note that each element of attribute vector might not have uniform distribution, e.g., the value for {\it age} is typically a number between 0 and 100. Therefore, each number may not have maximum randomness (entropy). However, even if each number has minimum randomness (1-bit entropy), for any input vector with $D\geqslant128$, the attack is not possible.

{\bf Compressive Sensing/Reconstruction Attacks.}
The theory of compressive sensing makes it possible to approximately recover $x$ from the hash embedding $E(x)$ given the random seeds used in $E(.)$. Thus, we need to ensure that neither users nor the server have any information about the random seeds used in the computation of $E(.)$. Every embedding $E(x) | x \in \mathcal{C}$, however, should be created using identical random seeds (see Appendix~\ref{sec:CS} for discussions).
We show that using secure function evaluation protocols, it is possible to create secure binary embeddings using the same set of random seeds while no-one knows the seeds used in the hashing function. We describe the solution in Section~\ref{sec:secure_hash}.
Since the generation of the hash embedding is a one-time operation, it is allowed to be costlier as it is independent of other query processes.

{\bf Multilateration/Correlation/Triangulation Attacks.}
Although, recovering $x$ from $E(x)$ is not possible without knowing the random seeds inside $h_i\ \forall i$, it is still possible to recover $x$ from $E(x)$ by combining a ``few'' calls to the function $E(.)$ over few known inputs $y_i's$ (similar to chosen-plaintext attack). The LSH property allows the estimation of any pairwise distance. Such estimations open room for ``triangulation'' attack which is hard to prevent. We explain the problem and the solution in Section~\ref{sec:secLSH}. This information leakage with LSH is one of the major reasons why sub-linear search with semi-honest clients and absence of trusted party is an open research direction.

We use a novel probabilistic transformation to show that converting the bits generated from LSH family into secure bits is suitable for public release in the semi-honest model since it is secure against triangulation attack. Our secure bits preserve only the near-neighbors in Hamming space, unlike LSH, do not allow estimation of all possible distances. Our final $l$-bit embedding functions will be denoted by $S(x)$ instead of $E(.)$ to signify the secure nature of $S(.)$.
Our solutions for making LSH secure is the main contributions of this paper, which makes sub-linear time PP-NNS possible in the semi-honest setting with no trusted party. In the process, we fundamentally leverage the theory of LSH from the privacy perspective.

Before we describe the technical details of our solution in Sections~\ref{sec:secLSH} and \ref{sec:secure_hash} respectively, we briefly give an overview of our final protocol.

\section{Proposed \sys{} Protocol for Sub-Linear Query Time PP-NNS}
The security of the final protocol is based on the proposed secure LSH (described in Section~\ref{sec:secLSH}). Utilizing Secure LSH, we can generate $l$-bit embeddings, for some $l$, $S(.)$, such $S(x)$ is safe for public release. Assuming that we know such embedding $S(.)$, our final protocol for sub-linear query time PP-NNS works in three phases:

{\bf 1. Generating Random Seeds of $S(.)$:}
This process needs to be performed only once and does not require any communication between servers (off-line).
Two servers are required in \sys{}. They need to generate random seeds (that are used in the black-box hash computation of $S(.)$ in phase two and three).
The final internal random seeds of $S(.)$ are generated inside the GC protoocl and is not known to anyone.
The mathematical detail of secure LSH embedding, $S(.)$, are described in Section~\ref{sec:secLSH} while the details on its black-box computation are described in Section~\ref{sec:secure_hash}.

{\bf 2. Constructing the Database (Hash Tables):}
Every data owner $x$ computes her $l$-bit secure binary embedding $S(x)$ using black-box hash computation by communicating to the servers. This $l$-bit signature $S(x)$ serves as the secure public identifier for client $x$.
Server \#1 which possesses all $S(x)$s, pre-processes the collection of $l$-bit binary strings $\{S(x): \ x \in \mathcal{C}\}$ and forms hash tables using the classical algorithm for sub-linear search~\cite{Book:Rajaraman_11}.

{\bf 3. Searching in Sub-Linear Time (Query Phase):} To find near-neighbors of point $x$, it suffices to find points $y$ such that the corresponding secure embeddings, $S(x)$ and $S(y)$, are near-neighbors in Hamming distance. Searching for close Hamming distance can be done very efficiently in sub-linear time using the well-known algorithms~\cite{Book:Rajaraman_11}.

It should be noted that other than the set $S_\mathcal{C} = \{S(x): \ x \in \mathcal{C}\}$, no information is transfered between clients and servers. Hence, if $S_\mathcal{C}$ is not sufficient to recover any of the client's information, the protocol is secure.
For better readability we summarize the end-to-end protocol in Figure~\ref{alg:HashPseudoAlgo}.

\section{The Key Ingredient: LSH Transformation}
\label{sec:secLSH}
We explain why traditional LSH (or any scheme) which allows for estimation of {\it any} pairwise distance is not secure in the HbC adversary model. We describe the attack followed by its solution. We later formalize the privacy budget.

\subsection{``Triangulation'' Attack}
\label{sec:traingulations}
To give more insight into the situation, we describe \emph{triangulation attack} which leads to an accurate estimation of any target client's attribute $q$. For illustration, we focus on two dimensions, but the arguments naturally extend in higher dimensions. Assume that we are given the LSH embedding $E(q)$ of the target point $q$ (instead of secure embedding $S(.)$). An attacker, who wants to know the attributes of $q$, can create three random data (points) in the space $A,\ B,\ $and$\ C$. Creating few random points is not hard, e.g., fake online profiles with random attributes. The protocol allows computation of their LSH embeddings $E(A), \ E(B),\ $and$ \ E(C)$.

Given the random points $A$, $B$, $C$, and their corresponding hashes $E(A)$, $E(B)$, and $E(C)$, the attacker can compute the Hamming distance between the hash values of $E(A)$, $E(B)$, and $E(C)$ with the target hash, $E(q)$. Given this information, the distances of $q$ with $A$, $B$, and $C$, denoted by $d_A$, $d_B$ and $d_C$, can be accurately estimated~\cite{Proc:Bayardo_WWW07}.

{\bf Estimation of Distances from LSH Embeddings: } Let us focus on estimating $d_A$ from $l$-bit binary LSH embedding $E(A)$ and $E(q)$. For illustrations let $l$ be equal to 5 and $E(A) = 11010$ and $E(q) = 10110$. Let $m$ be the measured number of bit matches between $E(A)$ and $E(q)$ out of $l$. For our case, we have $m=3$, because bits at locations 1, 4 and 5 of $E(A)$ and $E(q)$ are equal. Since every bit comes from an independent $1$-bit LSH scheme, we have $\mathbb{E}[n_{match}] = l\times Pr(h_i(q) = h_i(A))=m,$
where $n_{match}$ is the number of bit matches between two LSH embeddings and $\mathbb{E}[.]$ denotes the expected value of a random variable.

Thus we can estimate, in an unbiased way, the collision probability $Pr(h(A) = h(q))$ by the expression $\frac{m}{l}$, the mean number of bit matches. As we discussed in \sect{sec:LSH}, the collision probability is usually a monotonic function of the distance (or similarity) $Pr(h(A) = h(q)) = f(dist(A,q))$ where $f$ is the monotonic function. Every monotonic function has an inverse, thus
$$dist(A,q) = f^{-1}\big(\frac{m}{l}\big),$$
is an accurate estimator of the distance or similarity~\cite{Proc:Shrivastava_ECML12,Proc:Li_ICML14,Shrivastava:SOCC_16}. See Section~\ref{sec:traingluationDes} for details where we describe the implementation of triangulation attack.

After estimating the distances $d_A$, $d_B$ and $d_C$, the attributes of $q$ can be inferred using triangulation. Figure~\ref{fig:tri} shows a two-dimensional illustration of our setup. 
\begin{figure}[ht]
	\centering
	\includegraphics[width = 0.3\columnwidth]{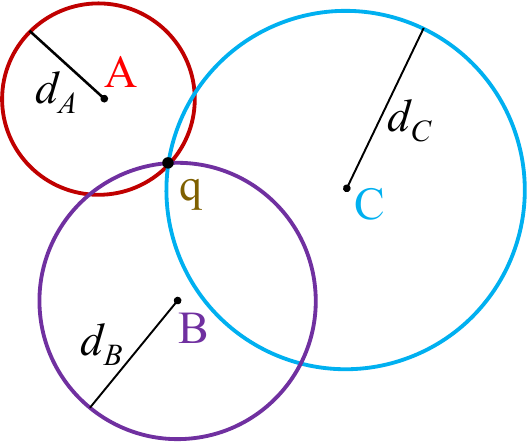}
	\caption{The user $q$ can be located using random points $A$, $B$, and $C$ along with the distances $d_A$, $d_B$, and $d_C$ which are estimated from the available binary embeddings (hashes).}
	\label{fig:tri}
\end{figure}

It should be noted that even if the distance estimation is not very accurate, generating many distance estimates from different random points would be sufficient to achieve a very good accuracy in locating any target point. 

The above illustration only shows two dimensions. For higher dimensions, we show an efficient iterative process, using the idea of alternating projections~\cite{boyd2003alternating}, to infer the attributes even for high-dimensional vectors. In Section~\ref{sec:traingluationDes}, we describe the process in details. Our inference process shows the power of simple iterative machine learning in breaking the security, which itself can be of independent interest. The ease of triangulation-based inference of attributes further emphasizes the need for more secure hashing schemes which we propose in the next section.

\subsection{Probabilistic Transformations for Generating Secure LSH}
Our proposal is a generic framework for making any given LSH privacy-preserving. In particular, we prevent LSH from leaking the distance information without compromising on the accuracy of the near-neighbor search.


We illustrate the main idea using 1-bit MinHash and later we formally introduce the methodology. The collision probability, for any two given data points $x$ and $y$, under 1-bit MinHash is given by $\frac{\mathcal{R}(x,y) + 1}{2}$ (Equation~\ref{eq:1bitprob}). This quantity varies linearly, between 1 to 0.5 as $\mathcal{R}(x,y)$ varies from 1 to 0, with a constant gradient of $\frac{1}{2}$.
Thus, even when $\mathcal{R}(x,y)$ is small, the variation of the collision probability with distance keeps changing and gets reflected in the Hamming distance between the public $l$-bit hash embeddings. This property allows us to estimate the distances accurately by counting the number of bit matches out of the $l$-bits which are public. For example, if 65\% of bits matches, then a good estimate of similarity between $x$ and $y$ is $0.65 \times 2 -1 = 0.3$ (Equation~\ref{eq:1bitprob}).

To make LSH privacy-preserving without losing the accuracy in near-neighbor search tasks, it is necessary to have the flat collision probability with no gradient if the similarity between the pair $x$ and $y$ is below the satisfactory level.
Thus, for any pair of random points $x$ and $y$, the Hamming distance between the publicly available $l$-bit hash codes is around $l/2$ (due to the 0.5 probability of agreement), which prohibits the estimation of distances between $x$ and $y$.

Until now, we have realized that we need to transform the collision probability.
The primary challenge is to find the precise expression for the curve which has the desired behavior and at the same time represents the collision probability of some $1$-bit hashing scheme. It should be noted that not every curve is a collision probability curve~\cite{Proc:Charikar}, therefore, it is not even known if such a mathematical expression exists.

We show that the expression given by $\frac{\mathcal{R}(x,y)^k + 1}{2}$, for some large enough $k$, has the required ``sweet'' property. In particular, we construct a new 1-bit secure MinHash with collision probability $\frac{\mathcal{R}(x,y)^k + 1}{2}$ for any positive integer $k$, instead of $\frac{\mathcal{R}(x,y) + 1}{2}$. The observation is that since $\mathcal{R} \le 1$, $\mathcal{R}^k$ for reasonably large $k$ quickly falls to zero as $\mathcal{R}(x,y)$ goes away from 1. Therefore, the quantity $\frac{\mathcal{R}(x,y)^k + 1}{2}$ will be very close to $\frac{1}{2}$ for even moderately high similarity. Furthermore, we can control the decay of the probability curve by choosing $k$ appropriately. The function $\frac{\mathcal{R}(x,y)^k + 1}{2}$ follows the desired trend of collision probability and is secure from information theoretic perspective.

The key mathematical observation is that we can generate $1$-bit hash functions with collision probability $\frac{\mathcal{R}(x,y)^k + 1}{2}$ by combining $k$ independent MinHashes. Note that, given $x$ and $y$, the probability of agreement of an independent MinHash value is $\mathcal{R}(x,y)$. Therefore, the probability of agreement of all $k$ independent MinHashes will be $\mathcal{R}(x,y)^k$, see~\cite{Proc:Shrivastava_NIPS13} for details. Also, to generate a 1-bit hash value from $k$ integers, we need a universal hash function that takes a vector of $k$ MinHashes and maps it uniformly to 1-bit. The final collision probability of this new $1$-bit scheme is precisely $\frac{\mathcal{R}(x,y)^k + 1}{2}$, as required. The overall idea is quite general and applicable to any LSH. We formalize it in the next section.

\subsection{Formalization}
\label{form}
As we argued in the previous section, we need a universal hashing scheme, $h_{univ}: \mathbb{N}^k \mapsto \{0,1\}$, which maps a vector of $k$ integers uniformly to 0 or 1. There are many ways to achieve this and a common candidate is
\begin{equation}
h_{univ}(x_1, \ x_2,..., \ x_k) = (r_{k+1} + \sum_{i =1}^{k} r_i x_i) \ mod\ p, \ mod \ 2,\notag
\end{equation}
where $r_i$ are fixed randomly generated integers.

Given a hash function $h$, uniformly sampled from any given locality sensitive family $\mathcal{H}$, let us denote the probability of agreement (collision) of hash values of $x$ and $y$ by $P_{c}$,
\begin{equation}
P_{collision} \equiv P_{c} \equiv {Pr}_\mathcal{H}(h(x) = h(y)).
\end{equation}

\begin{defn} [\emph{\bf Secure LSH}]
	Our proposed secure $1$-bit LSH, $h_{sec}$, parameterized by $k$, for any point $x$ is given by
	\begin{equation}
	h_{sec}(x) = h_{univ}(h_1(x),h_2(x),...,h_k(x)),
	\end{equation}
	where $h_i$s$, \ i \in \{1,\ 2,...,\ k\}$ are $k$ independent hash functions sampled uniformly from the LSH family of interest $\mathcal{H}$.
\end{defn}

It is not difficult to show the following:
\begin{thm}\label{th:collision}
	For any vectors $x$ and $y$, under the randomization of $h_{sec}$ and $r_i$ we have
	\begin{equation}\label{eq:pcsec}
	P_{c}^{sec} = Pr_{\mathcal{H},r}\left(h_{sec}(x) = h_{sec}(y)\right) = \frac{P_{c}^k + 1}{2}
	\end{equation}
\end{thm}

\begin{proof} It should be noted that $h_{sec}(x) = h_{sec}(y)$ can happen due to the random bit collision with probability $\frac{1}{2}$. Otherwise the two are equal if and only if $$(h_1(x),h_2(x),...,h_k(x))= (h_1(y),h_2(y),...,h_k(y)),$$
\noindent which happens with probability $P_{c}^k$, because each $h_i$ is independent and $Pr(h_i(x) = h_i(y)) = P_{c}$. Therefore, the total probability is $\frac{1}{2} + \frac{1}{2}P_{c}^k$ leading to the desired expression.
\end{proof}

We illustrate the usefulness of the framework proposed above in deriving secure $1$-bit hash for two most popular similarity measures: 1) Secure MinHash for Jaccard similarity and 2) Secure SimHash for Cosine similarity. The idea is applicable to any general LSH including ALSH for Maximum Inner Product Search (MIPS)~\cite{Proc:Shrivastava_NIPS14,Proc:Shrivastava_WWW15,Proc:Shrivastava_UAI15}.

\subsubsection{Making Minwise Hashing Secure (Secure MinHash)}
As an immediate consequence of Theorem~\ref{th:collision}, we can obtain secure 1-bit MinHash in order to search based on the Resemblance similarity,
\begin{equation}
h_{sec}^{min}(x) = h_{univ}(h_{\pi_1}^{min}(x),\ h_{\pi_2}^{min}(x),...,\ h_{\pi_k}^{min}(x)),
\end{equation}
with the following Corollary:
\begin{cor}
	For MinHash we have:
	\begin{equation}\label{eq:pcmin}
	P_{c}^{sec}=Pr\left(h_{sec}^{min}(x) = h_{sec}^{min}(y)\right) = \frac{\mathcal{R}^k + 1}{2}
	\end{equation}
\end{cor}
Figure~\ref{fig:r_to_powers} shows that the nature of new collision probability follows the desired trend. The parameter $k$ gives us the knob to tune the probability curve. In section~\ref{sec:practical}, we discuss how to tune this knob.

To generate our final $l$-bit binary embedding $S(x)$, we simply generate $l$ independent $h_{sec}^{min}$, by using independent permutations for MinHashes and independent random numbers for the universal hashing. Therefore, $S(x)$ is the concatenation of $l$ different $h_{sec}^{min}$.
\begin{figure}[t]
	\centering
	\mbox{\hspace{-0.12in}
		\includegraphics[width=0.5\columnwidth]{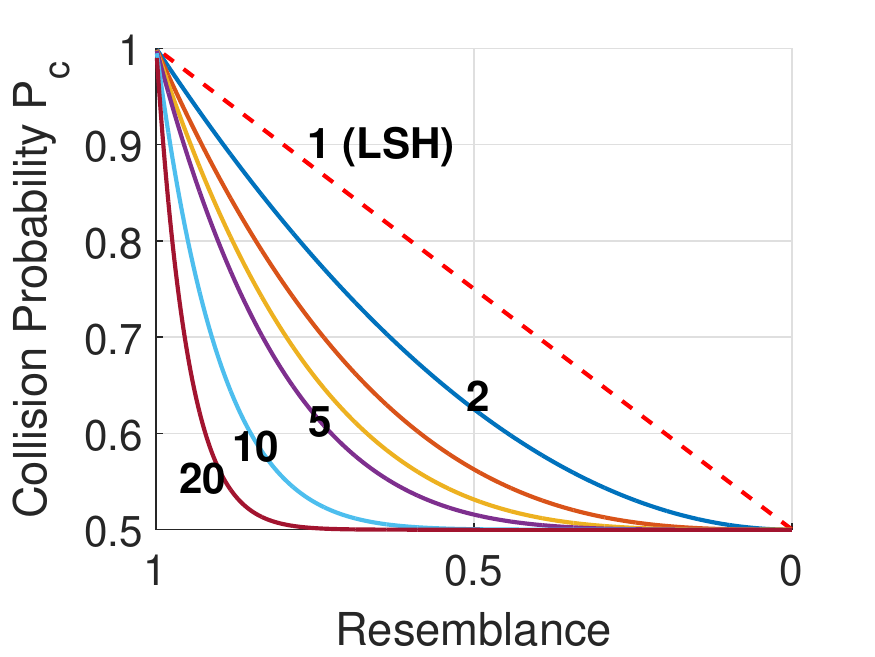}
		\includegraphics[width =0.5\columnwidth]{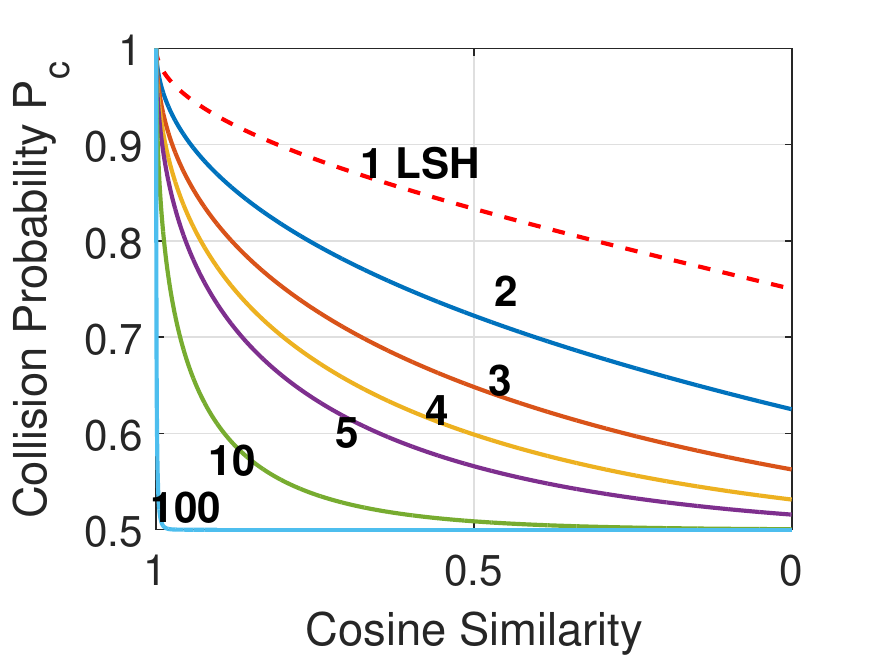}
	}
	\caption{{\bf Left:} The probability of collision of Secure MinHash as a function of R (resemblance) for different values of k. {\bf Right:} The probability of collision of Secure SimHash as a function of $\theta$ (Cosine Similarity) for different values of k. Increasing $k$ makes the collision probability drop to the constant 0.5 rapidly.}
	\label{fig:r_to_powers}
\end{figure}

In \sect{sssec:infoth}, we formally show that our transformed bits are more secure than LSH. In particular, we prove that the mutual information between the two secure 1-bit MinHashes, $h_{sec}^{min}(x)$ and $h_{sec}^{min}(y)$ decays sharply ({\it exponentially} with $k$) to zero as the similarity between $x$ and $y$ (i.e. $\mathcal{R}$) decreases. Thus, there is negligible mutual information about $x$ in the embedding of some random (non-neighbor) $y$.



\subsubsection{Making Signed Random Projections Secure (Secure SimHash)}
\label{sec:secSim}
Analogous to MinHash, we can make SimHash secure with the same properties.
\begin{equation}
h_{sec}^{sign}(x) = h_{univ}(h_{w_1}^{sign}(x),\ h_{w_2}^{sign}(x),...,\ h_{w_k}^{sign}(x)),
\end{equation}
where $w_i$s for all $i$ are independently chosen. Figure~\ref{fig:r_to_powers} (right) summarizes the collision probability as a function of similarity for different values of $k$.
\begin{cor}
	For Secure SimHash we have:
	\begin{align}\label{eq:pcsim}
	P_{c}^{sec}=Pr\left(h_{sec}^{sign}(x) = h_{sec}^{sign}(y)\right) = \frac{(1 - \frac{\theta}{\pi})^k + 1}{2}
	\end{align}
\end{cor}

\subsubsection{Info. Theoretic Bound as a Function of $k$}\label{sssec:infoth}
We provide the theoretical property of our transformation by quantifying the mutual information between two $l$-bit secure embeddings. The similarity of $x$ and $y$ ($Sim(x,y)$) is denoted as $S_{x,y}$.
\begin{thm}\label{thm:total_mutual_info}
       For any two data points $x$ and $y$, with $S_{x,y}$ being the similarity between them, the mutual information between $h_{sec}(x)$ and $h_{sec}(y)$ is bounded by
       \begin{equation}
       \label{eq:total_mutual_info}
       I(h_{sec}(x);h_{sec}(y)|S_{x,y})<
       l\cdot (2P_{c}^{sec}-1)log(\frac{P_{c}^{sec}}{1-P_{c}^{sec}})
       \end{equation}
\end{thm}
\begin{proof}
For simplicity let us call the $i^{th}$ bit of $h_{sec}(x)$, $u_i$ and $i^{th}$ bit of $h_{sec}(y)$, $u'_i$.  and derive the mutual information between these two bits conditioned on $S_{x,y}$ as follows:
\begin{align*}
I(u_i;u'_i|S_{x,y}) &\equiv \\
&\sum_{u_i,u'_i\in\{0,1\}} P(u_i,u'_i|S_{x,y}) log \frac{P(u_i,u'_i|S_{x,y})}{P(u_i|S_{x,y})P(u'_i|S_{x,y})} \\
&=P_{c}^{sec}log(2P_{c}^{sec})+(1-P_{c}^{sec})log(2(1-P_{c}^{sec}))\\
&<(2P_{c}^{sec}-1)log(\frac{P_{c}^{sec}}{1-P_{c}^{sec}})
\end{align*}
Since every bits of the binary embeddings are generated independently, the mutual information between $l$-bit embeddings is multiplied by $l$.
\end{proof}
Substituting $P_{c}^{sec}$ from Equation~\ref{eq:pcmin} and Equation~\ref{eq:pcsim}, the mutual information can be expressed as a function of Resemblance and Cosine similarities, respectively.
\begin{cor}
For secure MinHash we have:
\begin{align}\label{eq:imin}
I(h_{sec}^{min}(x);h_{sec}^{min}(y)| \mathcal{R}) <
\mathcal{R}^klog(\frac{1+\mathcal{R}^k}{1-\mathcal{R}^k})
\end{align}
and for Secure SimHash:
\begin{align}\label{eq:isim}
I(h_{sec}^{sign}(x);h_{sec}^{sign}(y)| \theta) < (1 - \frac{\theta}{\pi})^klog(\frac{1+(1 - \frac{\theta}{\pi})^k}{1-(1 - \frac{\theta}{\pi})^k})
\end{align}
\end{cor}
As can be seen from Equations~\ref{eq:imin} and~\ref{eq:isim}, the mutual information drops rapidly (exponentially with $k$) to zero for $x$ and $y$ that have small similarity. Thus, for any two non-neighbor points (small $Sim(x,y)$) the generated bits behave like random bits revealing no information about each other. Obviously, $k=1$, which is the traditional choice for LSH, is not secure, as the bits contain significant mutual information. The choice of $k$ controls the decay of the mutual information and hence is the privacy knob (see Section~\ref{sec:practical} for details on how to tune this knob).

\subsection{Formalism of Privacy Budget}
\label{sec:practical}
Suppose, the application at hand considers any pair of points $x$ and $y$ with $Sim(x,y) < s_0$ as non-neighbors, for some problem-dependent choice of $s_0$. The application also specifies an $\epsilon$ such that the collision probability of any two non-neighbors should not exceed $\frac{1}{2} +\epsilon$ (be very close to half (random)). Forcing this condition ensures that whenever $Sim(x,y) < s_0$, the released bits cannot distinguish $x$ and $y$ with any randomly chosen pair. Formally,

\begin{defn}[\label{def:priva}
	\emph{\bf $\epsilon$-Secure Hash at Threshold $s_0$}]
	For any $x$ and $y$ with $Sim(x,y) \le s_0$,
	we call a $1$-bit hashing scheme $h_{sec}$ secure at threshold $s_0$ if the probability of bit-matches satisfies:
	\begin{align}\notag
	\frac{1}{2} \le Pr(h_{sec}(x) &= h_{sec}(y)) \le \frac{1}{2} +\epsilon.
	\end{align}
\end{defn}
Note that the expression of $\epsilon$-secure hash is not symmetric since the probability of collision is always greater than or equal to $\frac{1}{2}$ (see Equation~\ref{eq:pcsec}).

We show that for any $\epsilon$-secure hash function, the mutual information in the bits of non-neighbor pairs is bounded.

\begin{thm}
[{\bf Information Bound}] For any $1$-bit $\epsilon$-secure hash function at threshold $s_0$, the mutual information between $h(x)$ and $h(y)$, for any pair with $Sim(x,y) \le s_0$, is bounded as
\begin{equation}\label{eq:infoSec}
I(h(x);h(y)) \le 2\epsilon\log{\frac{1+2\epsilon}{1 -2\epsilon}}
\end{equation}
\end{thm}
\begin{proof}
Follows from Theorem~\ref{thm:total_mutual_info} and Definition~\ref{def:priva}.
\end{proof}

In triangulation attack, we have access to $m$ attributes $y_i$s: $Y = {y_1, y_2,... ,y_m}$, and their corresponding hashes $h(y_i)$s. Assuming $y_i$s are independent, we can bound the mutual information about any target $x$ conditional on knowing $y_i$'s and $h_i$s as follows:
\begin{thm}
	 For any $1$-bit $\epsilon$-secure hash function at threshold $s_0$, the mutual information between $h(x)$ and \\
	 $\{h(y_1), h(y_2)... h(y_m)\}$, for any pair with $Sim(x,y_i) \le s_0$ and any pair of $y_i, y_j$ are independent, is bounded as
	\begin{equation}\label{eq:infoSec2}
		I(h(x);h(y_1) h(y_2)... h(y_m)) \le 2m\epsilon\log{\frac{1+2\epsilon}{1 -2\epsilon}}
	\end{equation}
\end{thm}
\begin{proof} 	
	Define subsets $\mathcal {T}\subseteq {\mathcal {V}}$, where $\mathcal {V} =n.$
	\begin{equation}\label{eq:infoSec2}
	\begin{split}
	I(h(x);h(y_1) h(y_2)... h(y_m)) = &I(h(y_1) h(y_2)... h(y_m); h(x)) \\
	 = & \sum_{T\subseteq {2,...,n}} (-1)^{|T|} I( T; h(x))\\
	 \le &2m\epsilon\log{\frac{1+2\epsilon}{1 -2\epsilon}}
	 \end{split}
	\end{equation}
	Note that the mutual information of any $y_i, y_j$ pair is 0 because they are independent.
\end{proof}
Thus, for small enough $\epsilon$, it is impossible to get enough information about any non-neighbor $x$ via triangulation. We verify this observation empirically in the experiments. Next, we show that Secure LSH can always be made $\epsilon$-secure hash function for any $\epsilon$ using an appropriate choice of $k$.
\begin{thm}
	Any Secure LSH, $h_{sec}$, is also an $\epsilon$-secure hash function at any given threshold $s_0$, for all
	\begin{equation}
	k \ge \bigg\lceil \frac{\log{2\epsilon}}{\log{\big(Pr(h(x) = h(y)| Sim(x,y) = s_0)\big)}} \bigg\rceil,
	\end{equation}
	where $\lceil . \rceil$ is the ceiling operation. Here, $h(x)$ is the original hash function from which the $h_{sec}$ is derived.
\end{thm}
\begin{proof}
Follows from the definition of $\epsilon$-secure hashing added with fact that $h(x)$ satisfies Definition~\ref{def:lsh}.
\end{proof}

In order to obtain $\epsilon$-secure MinHash, we need $k = \bigg\lceil \frac{\log{2\epsilon}}{\log{s_0}} \bigg\rceil$. For secure SimHash, we need to choose $k = \bigg\lceil \frac{\log{2\epsilon}}{\log{\big(1- \frac{cos^{-1}(s_0)}{\pi}\big)}} \bigg\rceil$. To get a sense of quantification, if we consider $s_0 = 0.75$ (high similarity) and $\epsilon = 0.05$, then we have $k=8$ (MinHash) and $k=12$ (SimHash). 

\subsection{Utility-Privacy Trade-off of Secure LSH}
As mentioned in Section ~\ref{lsh}, the querying time and space for approximate Near-Neighbor search are directly quantified by $\rho = \frac{\log p_1}{ \log p_2} < 1$. The space complexity grows as $n^{1+\rho}$, while the query time grows as $n^{\rho}$, where $n$ is the size of the dataset. Thus, smaller $\rho$ indicates better theoretical performance.
The collision probability of our secured LSH is $P_{c}^{sec} = \frac{P_{c}^k + 1}{2}$. The new $\rho'$ for Secure LSH would be $\frac{\log \frac{P_{1}^k + 1}{2}}{\log \frac{P_{2}^k + 1}{2}}$.
\begin{thm}
	$\rho'$  is monotonically increasing with $k$.
\end{thm}
\begin{equation}
	\frac{d \rho'}{dk} =  \frac{p_1^k\ln\left(p_1\right)}{\ln\left(\frac{p_2^k+1}{2}\right)\left(p_1^k+1\right)}-\dfrac{\ln\left(\frac{p_1^k+1}{2}\right)p_2^k\ln\left(p_2\right)}{\ln^2\left(\frac{p_2^k+1}{2}\right)\left(p_2^k+1\right)} >0
\end{equation}
Therefore, when we increase $k$, we get the privacy at the cost of reduced space and query time. The quantification of this tradeoff is given as $\frac{\log \frac{P_{1}^k + 1}{2}}{\log \frac{P_{2}^k + 1}{2}}$

\section{Hiding the Mechanism of S(.)}
\label{sec:secure_hash}
We are now ready to describe the final piece of our protocol. We describe in detail how we can reasonably hide the random seeds inside of $S(.)$ from the users in addition to {\it both servers}. To compute the hash of the client's input, we need random seeds (e.g., for MinHash, random seeds are the random permutations and for SimHash, they are random vectors used in the projection step). Since these random seeds should be equal for all clients, we cannot let each client generate her seeds independently. Seeds should be chosen in a consistent fashion. However, seeds should not be revealed to any server, otherwise, they might be used to reconstruction the secret attributes.
As a result, we need to design a mechanism such that no party knows the seeds, which is an important and yet difficult task.
To compute $S(x)$ securely without revealing the actual random seeds to any party, at least two different (non-colluding) servers need to be deployed. While we do not trust either server, we require that two servers do not collude.


\begin{figure}[ht]
       \centering
       \includegraphics[width=\columnwidth]{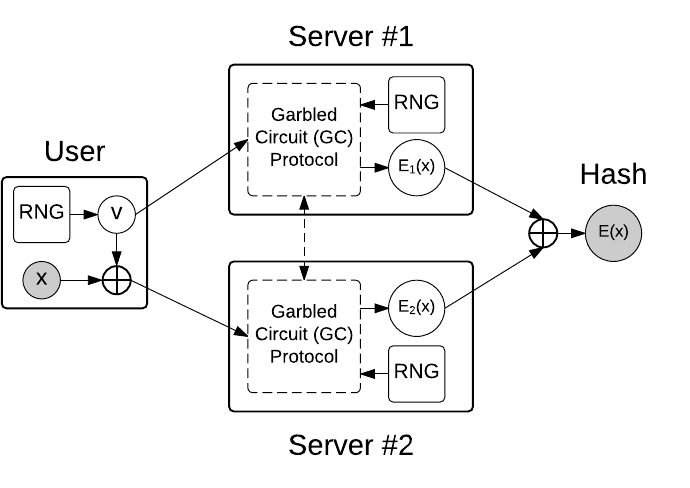}
       \caption{Global flow of black-box hash computation. RNG stands for Random Number Generator.
       } 
       \label{fig:Arch}
\end{figure}

In the initial phase, each server generates its own version of the seeds randomly. Whenever a client wants to compute a secure hash, $S(x)$, she generates a random $D$ dimensional vector $v$ (same dimensionality as her input value) and then XORs this vector with $x$ (resulting in $x \oplus v$). She then sends $v$ to server \#1 and $x \oplus v$ to server \#2. The term $x \oplus v$ is One-Time Pad (OTP) encryption of $x$ using $v$ as the pad and we denote it as $Enc_v(x)$. Given this information and the two initial random seeds, then both servers engage in a two-party secure computation. Here, we utilize Garbled Circuit (GC) protocol in order to compute $S(x)$.
The GC protocol is one of the generic secure function evaluation protocols that allows two parties to jointly compute a function on their inputs while keeping each input private to their respective owners. In this protocol, the function that is evaluated securely has to be described as a Boolean circuit. The computation and communication complexity of this algorithm is proportional to the number of non-XOR gates in the circuit.

The global flow of our approach is illustrated in Figure~\ref{fig:Arch}.
Server \#1 inputs $v$ and server \#2 inputs $Enc_v(x)$ ($x \oplus v$) to the GC protocol. In addition, each server inputs her random seeds to the GC protocol. {\it Actual seeds} used for generating the hash of $x$ are based on the two random seeds from two servers and are generated using the Boolean circuit that is used inside the GC protocol. In our case, the Boolean circuit is the secure hash computation suggested in Section~\ref{sec:secLSH}. For this reason, we have designed the corresponding Boolean circuits for securely computing secure MinHash and SimHash. The internal architecture of the two circuits are described in Appendix~\ref{ssec:circuit_minhash} and Appendix~\ref{ssec:circuit_simhash}, respectively.
After two servers run the GC protocol, they both acquire secret shared values of the hash ($S_1(x)$ and $S_2(x)$) and server \#1 needs to XOR the two values to get the real hash ($S(x)$). RNG stands for Random Number Generator.
The security proof of our proposed approach is given in Proposition~\ref{th:approachSecure} in Appendix~\ref{sec:sbhc}.
We utilize the recent advances which make hashing algorithmically faster~\cite{Proc:OneHashLSH_ICML14,Proc:Shrivastava_UAI14}.



The above procedure is called {\it XOR-sharing} technique and is secure in HbC attack model because: (i) server \#1 receives nothing but a true random number which contains no information about $x$ and (ii) server \#2 receives the encryption of the message $x$ using $v$ as the encryption pad and is perfectly secure~\cite{paar2009understanding}. Since both servers are assumed to not collude, they cannot infer any information about the user's input $x$.
The theory behind the GC protocol guarantees that neither of the parties that execute the protocol can infer any information about the intermediate values~\cite{yao1986generate}. Since the actual random seeds used to compute $S(x)$ is created by the GC protocol as an intermediate value, none of the servers nor the users know the value of true seeds and hence our protocol is secure.

\section{Noise Addition Methods and Their Poor Utility-Privacy Trade-off}\label{sec:noisePriv}
Obfuscating information by adding noise is one of the most popular techniques for achieving privacy. By adding sufficient noise to the hashes, one can construct $\epsilon$-secure scheme satisfying Definition~\ref{def:priva}. However, any protocol based on adding a noise will obfuscate the information uniformly in every bit, which will significantly affect the utility of near-neighbor search. We elaborate this poor utility-privacy trade-off. This is not the first time when such poor utility-privacy trade-off is being observed by adding a noise~\cite{fredrikson2014privacy}.

Following popular noise addition mechanism~\cite{kenthapadi2012privacy}, in order to achieve the requirement in Definition~\ref{def:priva}, we can choose to corrupt $1$-bit LSH $h(x)$ with a random bit, with probability $f$. Formally, the generated hash function is
$$h_{corr}(x) = \begin{cases}
random\_bit, \ \ \text{with probability f}\\
h(x), \ \ \ \ \ \ \ \ \ \ \ \text{with probability 1-f}
\end{cases}$$

\begin{thm}
	The new collision probability after this corruption, for any $x$ and $y$, is given by:
	\begin{align}\label{eq:noiseCollProb}\notag
	P(h_{corr}(x) &= h_{corr}(y)) \\ &= (1 - f)(Pr(h_{1bit}(x) = h_{1bit}(y)) + \frac{f}{2}.
	\end{align}
\end{thm}
Let us define $P(s) = Pr(h_{1bit}(x) = h_{1bit}(y) | Sim(x,y)= s)$. Using this quantity, it is not difficult to show:
\begin{thm}
	$h_{corr}$ is $\epsilon$-secure at threshold $s_0$, iff
	\begin{align}\notag
	(1 - f)P(s_0) + 0.5f \le 0.5 + \epsilon; \  \  \
	f \ge 1 - \frac{\epsilon}{\bigg(P(s_0) -\frac{1}{2}\bigg)}.
	\end{align}
\end{thm}
For $1$-bit MinHash with corruption, the collision probability boils down to $\frac{R(1-f)+1}{2}$. Thus, $f$ only changes the slope of collision probability curve.
To ensure $\epsilon$-secure hash at similarity $s_0$ threshold, we must have
\begin{align}
f &\ge 1 - \frac{2\epsilon}{s_0}.
\end{align}

\begin{figure*}[t]
       \centering
       \mbox{\hspace{-0.12in}
              \includegraphics[width=0.5\columnwidth]{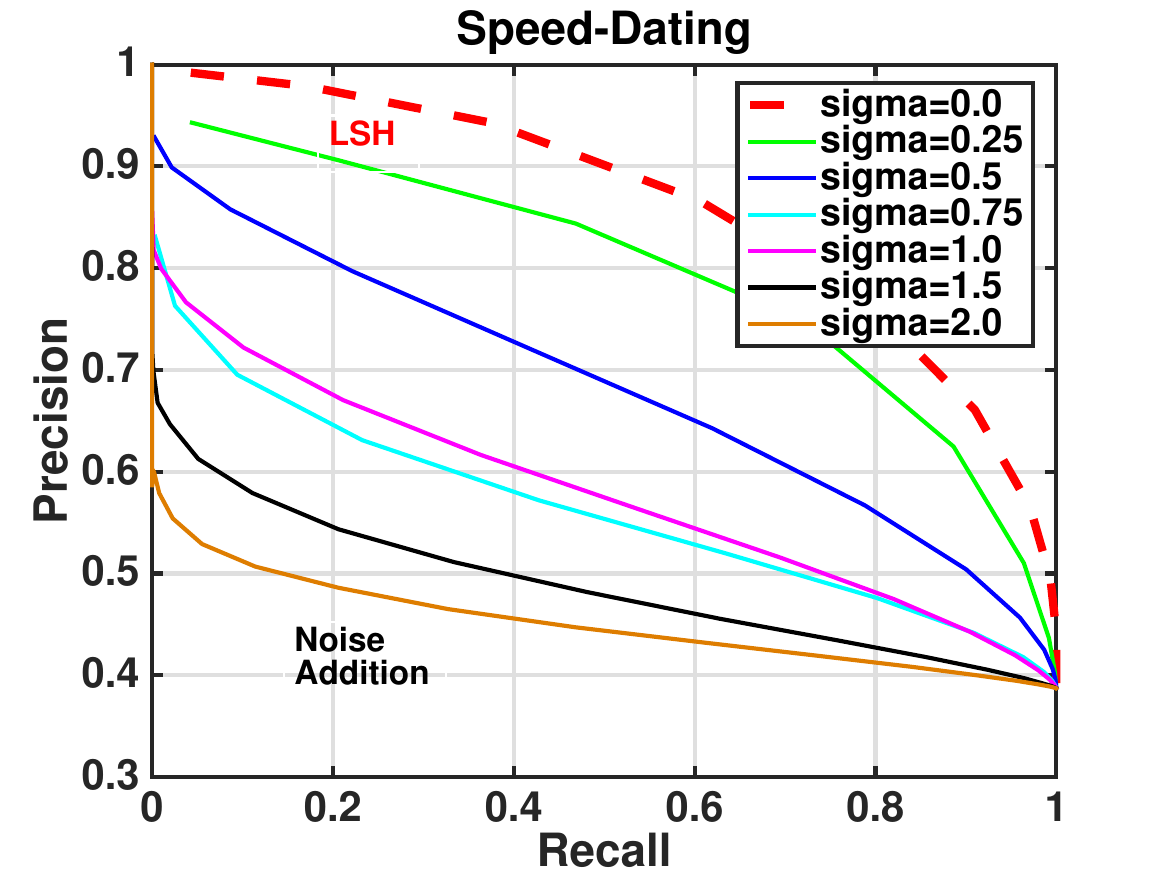}
              \includegraphics[width=0.5\columnwidth]{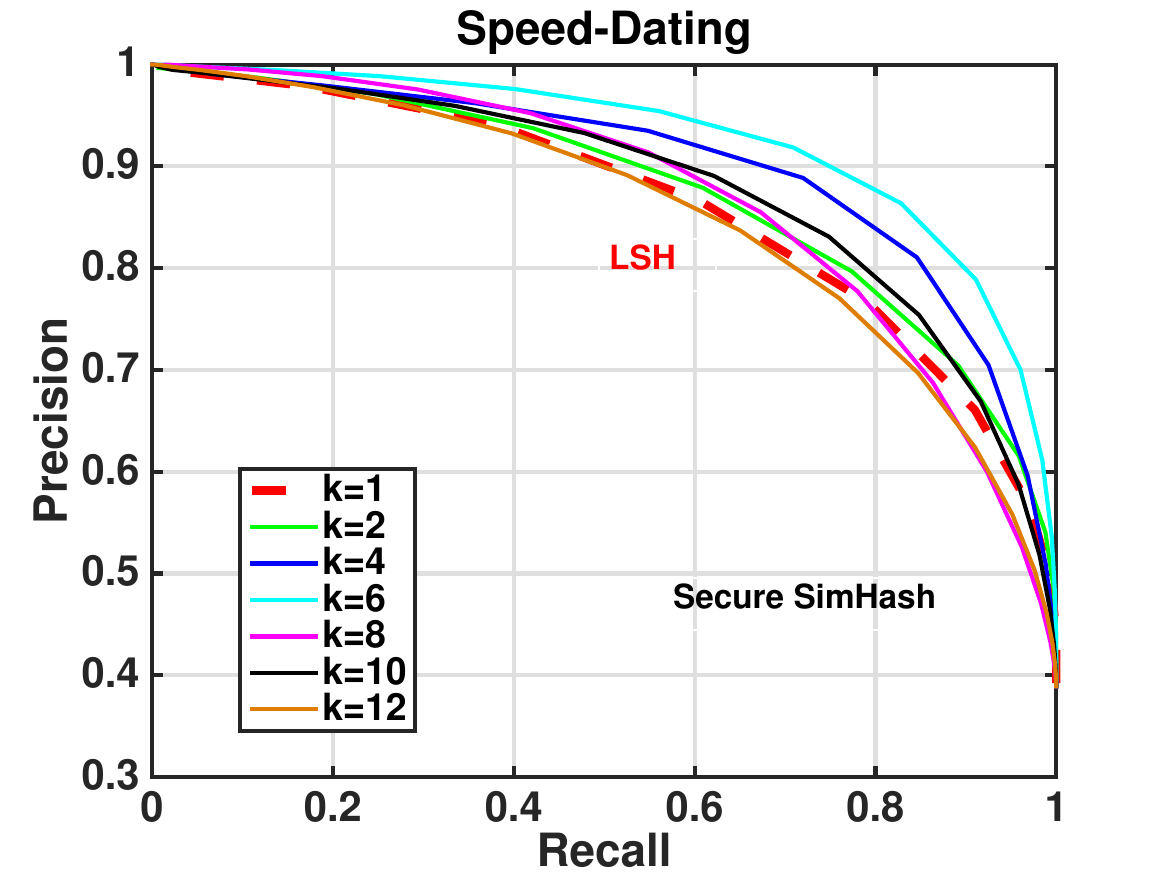}
       }
       \mbox{\hspace{-0.12in}
              \includegraphics[width=0.5\columnwidth]{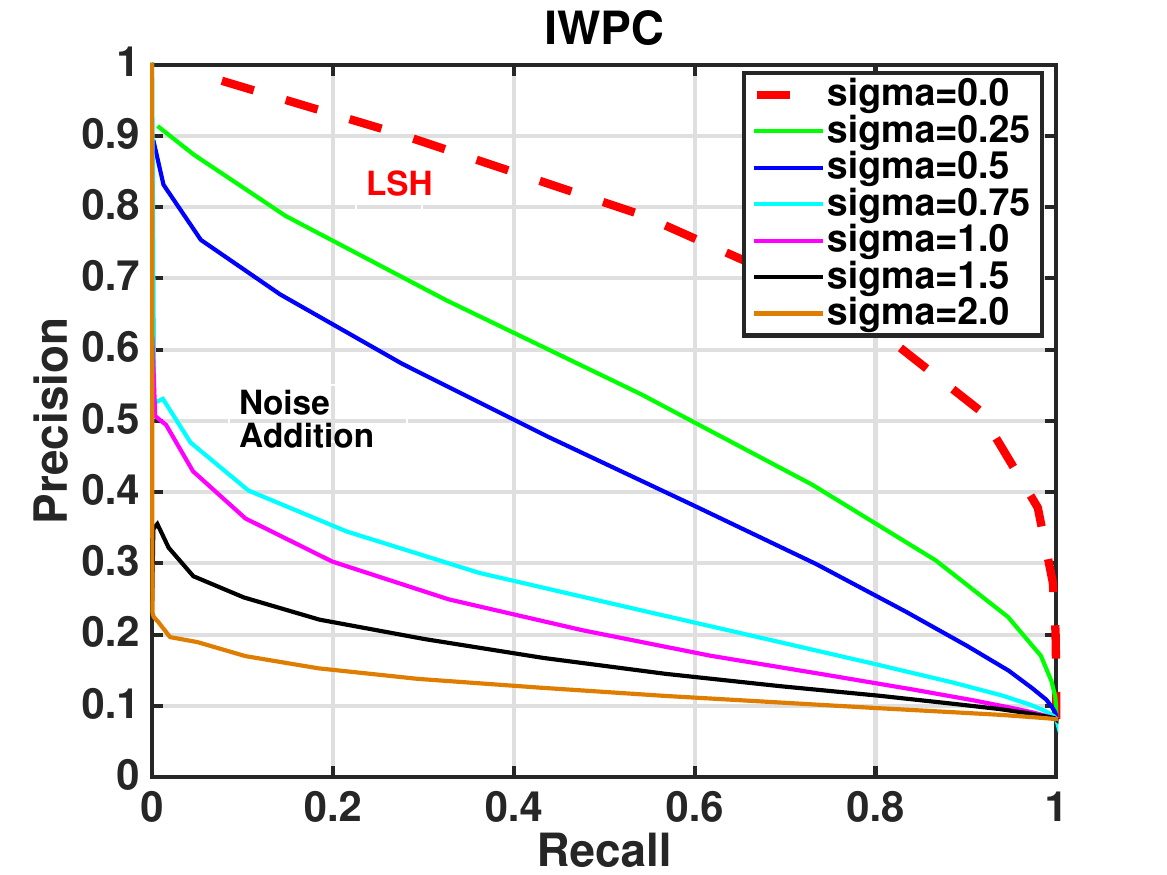}
              \includegraphics[width=0.5\columnwidth]{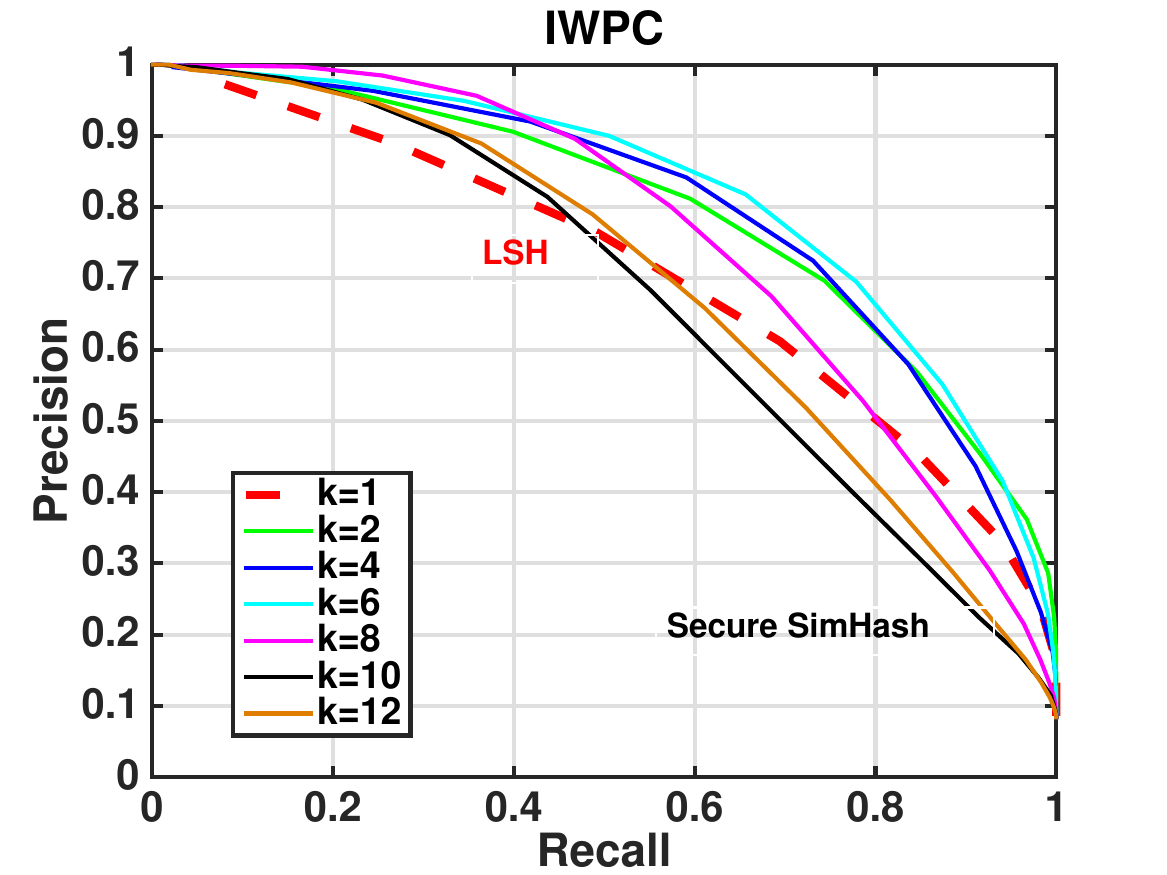}
       }
       \caption{ {\bf Utility-Privacy Tradeoff:} The plots represent the precession recall curves (higher is better) based on noise addition (first and third from the left) and secure cosine similarity (second and fourth from the left) for both datasets. The dotted red line is the vanilla LSH. We can clearly see that adding noise loses utility while the proposed approach is significantly better.}
       \label{fig:utilitytrade} 
\end{figure*}

To understand its implication, consider, an example with $s_0 = 0.75$ (high similarity) and $\epsilon = 0.05$. This combination requires $f \ge 0.86$. Such high $f$ implies that most bits ($86\%$) are randomly chosen, and hence they are uninformative. Even for very similar (almost identical) $x$ and $y$, the collision probability is close to random.
This degrades the usefulness of LSH scheme significantly.
In contrast, for the same threshold $s_0=0.75$ and same epsilon $\epsilon=0.05$, secure LSH needs $k=8$ which leads to the collision probability expression $\frac{\mathcal{R}^8 + 1}{2}$. For $x =y$, i.e. $\mathcal{R} =1$, this expression is {\it always} 1. For $x$ and $y$ with similarity $0.95$, the collision probability is greater than $0.83$, significantly higher than $0.56$ obtained using noise addition (very close to the random probability $0.5$).


\section{Evaluations}
\label{sec:eval}

\subsection{Utility-Privacy Tradeoff}
In this section, we provide thorough evaluations of the accuracy and privacy trade-off in our framework. Our aim is two-fold: (i) We want to evaluate the benefits of our proposal compared to traditional LSH in preventing triangulation attack and simultaneously evaluate the effect of our proposal on the utility of near-neighbor search. (ii) We also want to understand the utility-privacy of noise addition techniques in practice and further quantify it with the trade-offs of our approach. It is important to have such evaluations, as pure noise addition may be a good heuristic on real datasets that prevents the triangulation attack without hurting accuracy.

{\bf Datasets:} We use the IWPC~\cite{international2009estimation} and Speed Dating datasets~\cite{fisman2006gender}. They belong to different domains but both contain private and sensitive attributes of the concerned individuals. The IWPC is a medical dataset which consists of 186 demographic, phenotypic, and genotypic features like race, medicines taken, and Cyp2C9 genotypes of 5700 patients. We split the records to 80\% for creating hash tables and 20\% for querying.
The dataset is publicly available for research purposes.
The type of data contained in the IWPC dataset is equivalent to that of other private medical datasets that have not been released publicly~\cite{fredrikson2014privacy}. Speed-Dating dataset has 8378 text survey samples, each has 190 features representing geometric features or answers to designed questions by subjects.

We focus on the cosine similarity search, therefore, our underlying LSH scheme is SimHash (or Signed Random Projections). The gold standard neighbors for every query were chosen to be the points with cosine similarity greater than or equal to 0.95. Please note that LSH is threshold-based~\cite{Proc:Indyk_STOC98}. Hence, we chose a reasonable high similarity threshold.

{\bf Baselines:} We chose the following three baselines for our comparisons.
{\bf 1. LSH:} This is the standard SimHash-based embedding.
{\bf 2. Secure LSH:} As described in Section~\ref{sec:secLSH}, we use our proposed transformation to make LSH secure. To study the utility-privacy trade-off, a range for privacy parameter $k = 2,4,6,8,12$ is chosen. Note, $k=1$ is vanilla SimHash.
{\bf 3. Noise-based LSH:} ~\cite{kenthapadi2012privacy} shows a way to release user information in a privacy-preserving way for near-neighbor search. The paper showed that adding Gaussian noise $N(0, \sigma^2)$ after the random projection preserves differential privacy. To compute the private variant of SimHash, we used the sign of the differentially random private vector (generated by perturbed random projections) as suggested in~\cite{kenthapadi2012privacy}. To understand the trade-off the noise levels are varied over a fine grid $\sigma=0, 0.25, 0.5, 0.75, 1.0, 1.5, 2.0$.

We generated 32-bit hashes for IWPC and 64-bit hashes for Speed-Dating using each of the competing candidate hashing schemes. For each query data, we ranked points in training data based on the Hamming distance of the competing hash codes. We then computed the precision and recall of the Hamming-based ranking on the gold standard neighbors. We summarized the complete precision-recall curves for both the datasets and all the competing scheme in Figure~\ref{fig:utilitytrade}. This is a standard evaluation for hashing algorithms in the literature~\cite{weiss2009spectral}. Higher precision-recall under a given ranking indicates a better correlation of binary Hamming distance with the actual similarity measure. A better correlation directly translates into a faster algorithm for sub-linear near neighbor search~\cite{Book:Rajaraman_11} with Hamming distance.

\begin{figure*}[t!]
	\centering
	\mbox{\hspace{-0.12in}
		\includegraphics[width=0.5\columnwidth]{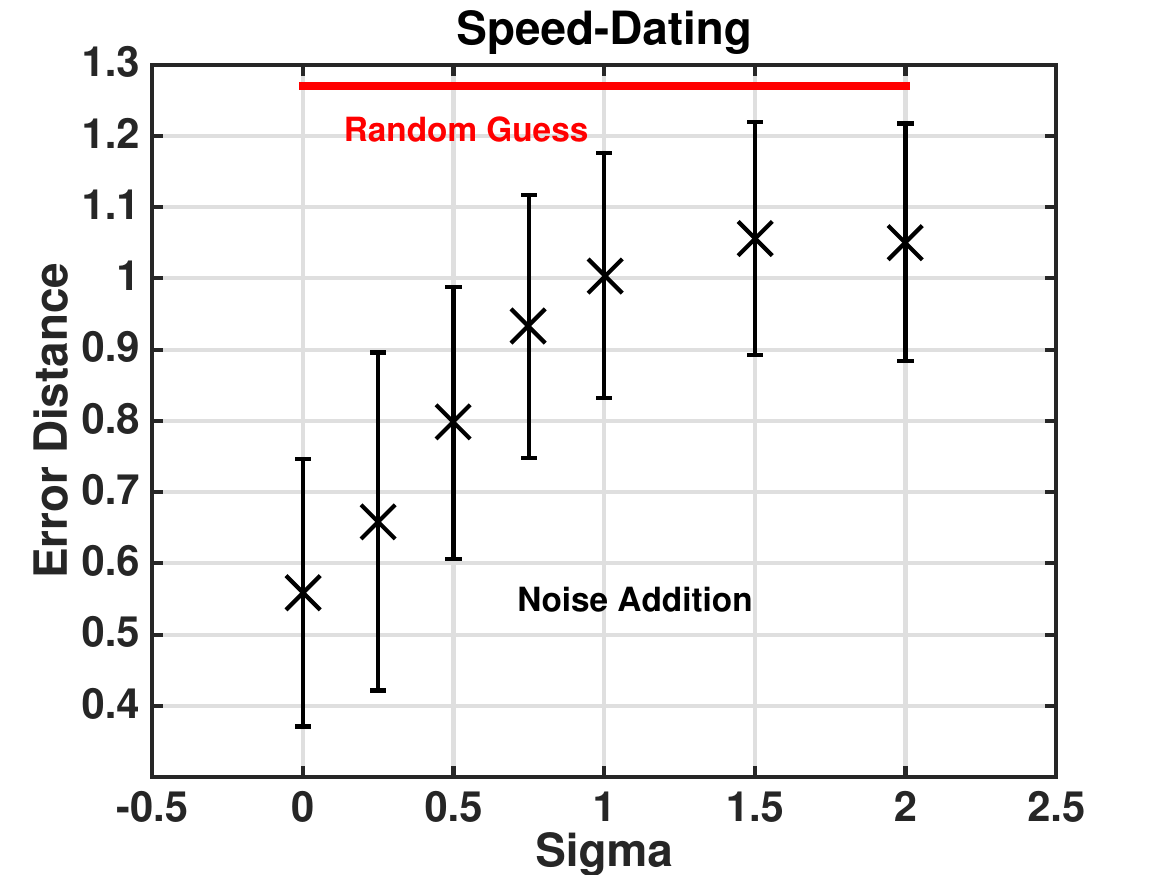}
		\includegraphics[width=0.5\columnwidth]{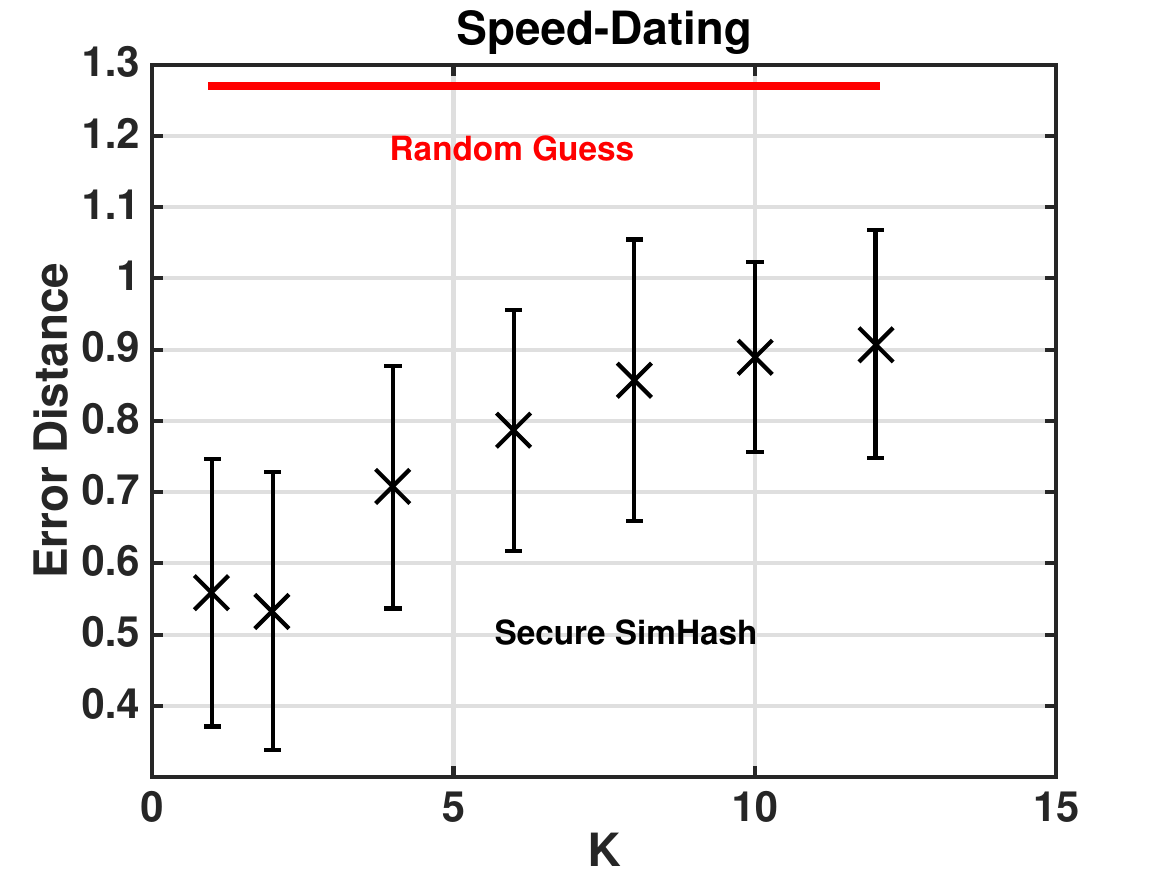}
	}
	\mbox{\hspace{-0.12in}
		\includegraphics[width=0.5\columnwidth]{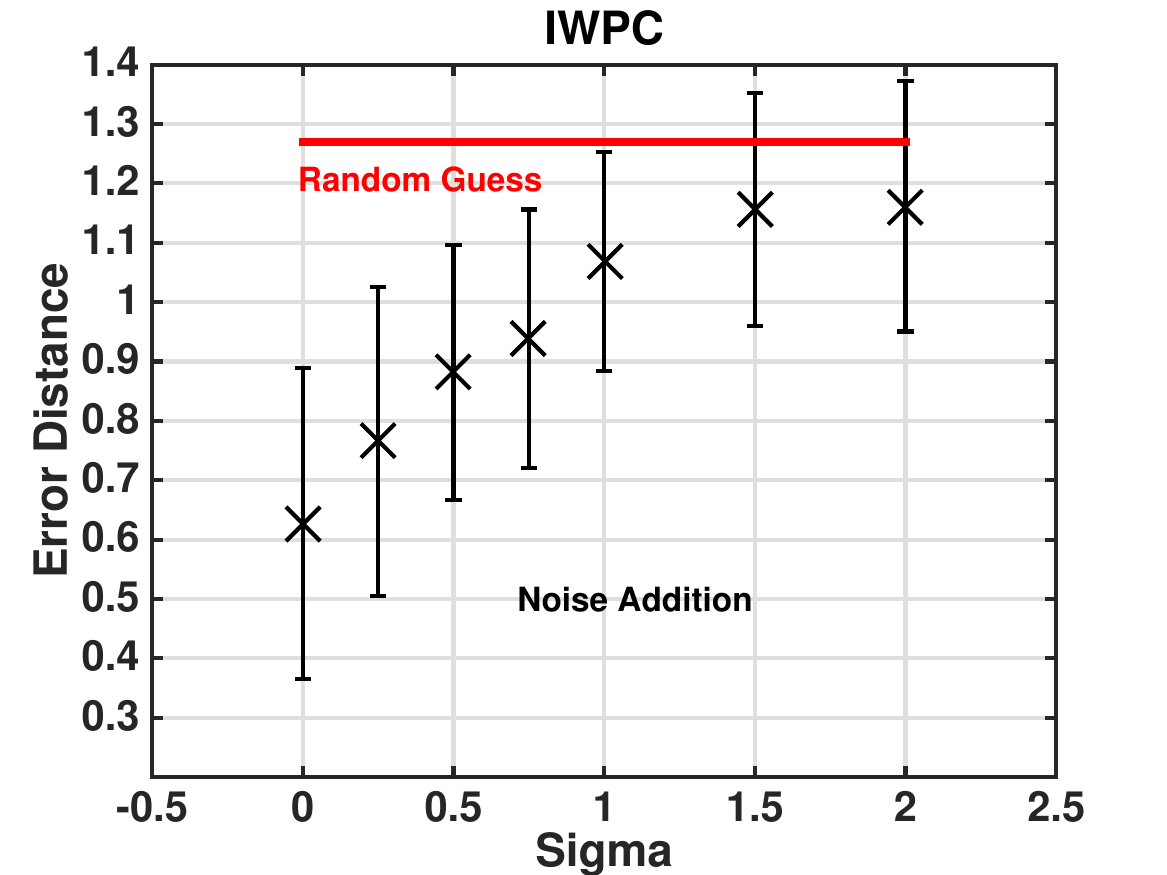}
		\includegraphics[width=0.5\columnwidth]{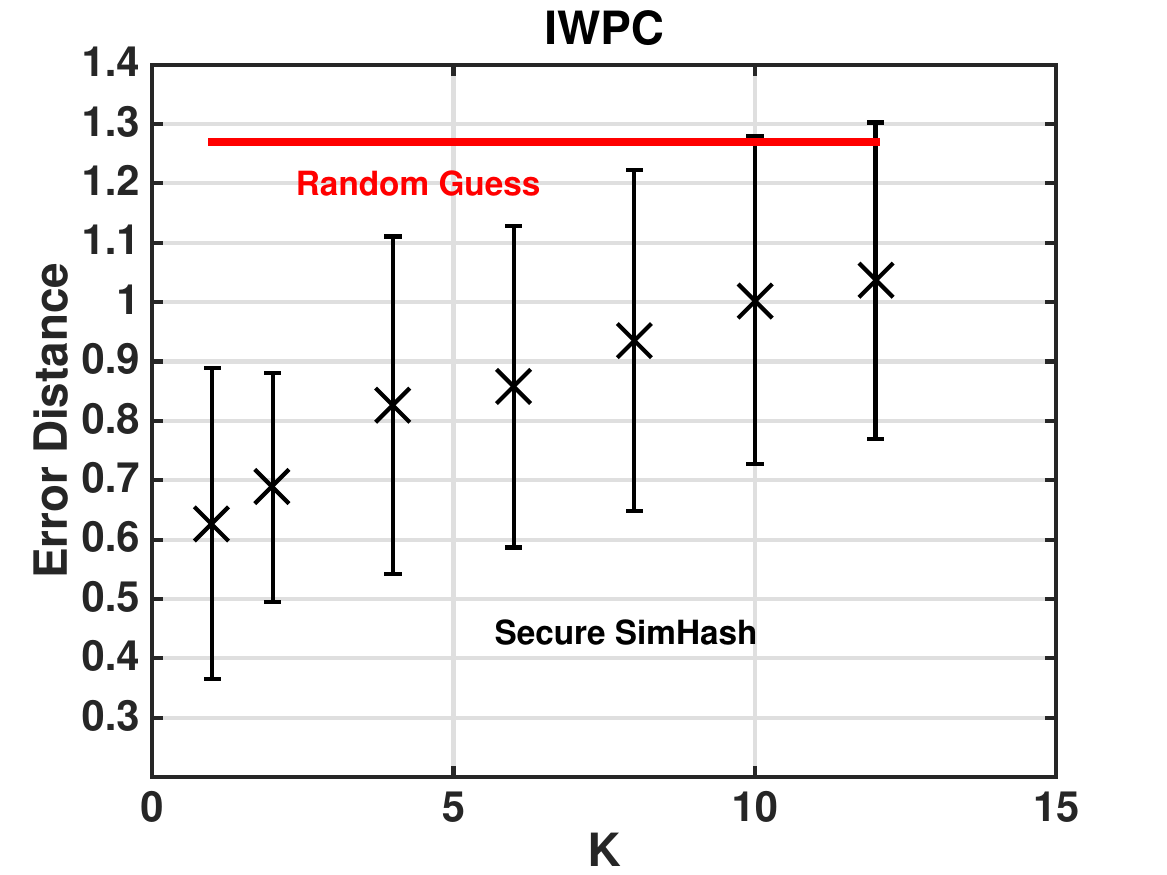}
	}
	\caption{ {\bf Effectiveness Against Triangulation Attack:} Plots show the error (mean and error bars) in the triangulation-attack-based inference of attributes (higher is more secure, random is holy grail). We can see that both adding noise (first and third form the left) and increasing $k$ with secure hashing (second and fourth form the left) lead to increased security. Contrasting this with Figure~\ref{fig:utilitytrade} clearly shows the superiority of our proposal in retaining utility for the same level of privacy.}
	\label{fig:privacytrade}
\end{figure*}

In Figure~\ref{fig:utilitytrade}, the first and third plots from the left-hand side show the retrieval precision and recall curve using various $\sigma$ in noise-based LSH. The Vanilla LSH line, which is the performance of LSH $k=1$ or $\sigma=0$ serves as the reference in the plots. By increasing $\sigma$, the accuracy of noise-based hashing drops dramatically. Adding noise as argued in Section~\ref{sec:noisePriv} leads to poor collision probability for similar neighbors which in turn leads to a significant drop in accuracy compared to LSH as evident from the plots. As the privacy budget is increased, by adding more noise, the performance drops significantly. In contrast, the second and fourth plots from the left-hand side show the precision and recall curve using different $k$s with Secure SimHash.

By increasing privacy budget $k$, the accuracy does not drop and even gets better than vanilla LSH. This improvement is not surprising and can be attributed to the enhanced gap between the collision probability of near-neighbor and any random pair (Figure~\ref{fig:r_to_powers}). It is known that with hashing-based techniques, such enhanced gap leads to a better accuracy~\cite{Proc:Indyk_STOC98}. The plots show a consistent trend across the datasets and clearly signify the superiority of our proposed transformation over both LSH and noise addition based methods in terms of retrieving near-neighbors. The result clearly establishes the importance of studying problem-specific privacy before resorting to obfuscation based on noise.

\subsection{Effectiveness Against Triangulation Attack}
We showed that irrespective of the privacy budget, our proposal is significantly more accurate than LSH and Noise-based LSH. Our theoretical results suggest that the proposal is also secure against triangulation attack, whereas, vanilla LSH is not. We validate this claim in this section. Furthermore, we also study the effectiveness of noise addition in preventing the attack.

To evaluate the privacy, we implemented the ``triangulation attack'' and inferred its accuracy on real datasets, IWPC, and Speed-Dating. The task was to infer sensitive attributes of a given target query vector by triangulating it with respect to randomly chosen points as explained in Section~\ref{sec:traingulations}. For IWPC, we selected some sensitive attributes for inference like cancer or not, set of medicines taken, or Cyp2C9 genotypes to form the attack data points. For Speed-Dating, we randomly chose the attributes for inferring. To scale-up the implementation for higher dimensions, we use a novel iterative projection algorithm which successively approaches the target. The procedure is described separately in Section~\ref{sec:traingluationDes}, which could be of separate interest.

We used the same privacy budget, i.e., $k = 1,2,4,6,8,12$ for Secure SimHash and $\sigma=0, 0.25, 0.5, 0.75, 1.0, 1.5, 2.0$, for noise-based SimHash. Again, $k=1$ and $\sigma=0$ corresponded to the vanilla LSH method which will serve as our reference point. We computed the error of the estimated target using triangulation attack with the actual target. We then calculated the mean and standard deviations of the errors over 100 independent triangulation attacks. The errors for varying $k$ for our proposed secure LSH and varying $\sigma$ for noise-based LSH were summarized in Figure~\ref{fig:privacytrade}. We also plotted the accuracy of random guess which will serve as our holy grail for privacy. The attack accuracy for $k=1$ ($\sigma=0$) is substantially better than the random guess which clearly indicates the vanilla LSH is not secure. The decrease in attack accuracy with an increase in $k$ clearly shows the high security level of our solution.

As indicated by our theoretical results, the accuracy of the triangulation attack decreases and slowly approaches the random level (holy grail for privacy) as the privacy budget increases. We can conclude that both noise addition and our proposal effectively prevent triangulation attack. Increasing noise, as expected, preserves privacy but at a significant loss in utility. However, the retrieval experiments show that our proposal provides privacy {\it without} loss in accuracy. For all $\sigma$, there always exists some $k$ which could achieve significantly better performance for the same level of security.

\subsection{Computational Cost Comparison with SFE Protocols}\label{ssec:cwsfe}
In this section, we compare the performance of our protocol with the GC protocol, one of the most promising and efficient Secure Function Evaluation (SFE) protocols. In our scheme, we have integrated the GC protocol only for our black-box hash computation step that is computed independently and only once for each client. We will compare the performance of our protocol with the {\it pure execution of NNS} in GC to show the shortcomings of this approach. While GC protocol can compute NNS without any computational error (compare to Figure~\ref{fig:utilitytrade}), it has rather limited practical usage and scalability.
The recent work of~\cite{songhori2015compacting} has implemented the K-Near-Neighbor (KNN) search using TinyGarble~\cite{songhori2015tinygarble} framework, one of the most efficient GC frameworks. Based on their performance results, they report execution time of 6.7s for $N=128,000$ when processing 31-bit data.
According to their cost functions (which scales linearly with $N$ and input bit length), for $N=3$ Billion and input size of 1280-bit (same as ours), the execution time exceeds {\it 74 days}. In contrast, our protocol requires $0.415$ second for black-box hash computation
and $0.887$ second to search the hash-tables, resulting in an overall $1.3$ s execution time on the same machine.
We have also modified their solution and synthesized the circuit for NNS based on the Cosine similarity (see Appendix~\ref{sec:newexp}).
For the exact same problem and parameters as ours, their solution requires an estimated processing time of $1.5\times 10^{8}$ seconds and communication of $1.2\times 10^{7}$ GBytes. This clearly illustrates the superiority of our novel scheme over GC.

\section{Alternating Projections for Triangulation Attack}\label{sec:traingluationDes}
We provide the details of our implementation for the triangulation attack over SimHash with cosine similarity (angles) as the measure. We start with all normalized vectors. Given the target point $q$, we generate $D+1$ random points $X_i$s in the space.

\begin{align}
\label{eq:a}
q \in R^{D}, X_{i} \in R^{D}, \lVert X_{i} \rVert_2 &= \lVert q \rVert_{2} = 1,\\ \notag
&\forall i \in \{1, 2,..., (D+1)\}.
\end{align}

\noindent The distance between every $X_{i}$ and $q$,
\begin{equation}
\label{eq:b}
d_{i} = \lVert X_{i}-q \rVert_2, \ \  \forall i \in \{1, 2,..., (D+1)\}
\end{equation}
\noindent is estimated as described in Section~\ref{sec:traingulations}, first we estimate the angle $\theta$ using hash matches between $H(X_i)$ and $H(q)$: Then, we can get the distance $d_i$, from $\theta$ easily as the data is normalized. 

After finding all of the distances, we use Alternating Projection Method~\cite{gubin1967method} to find the possible intersection of $D+1$ $D$-dimensional spheres, $S_1, ..., S_{D+1}$, each with central point $X_i$ and radius $\lVert X_{i}-q \rVert_2$.  Any point in the intersection will likely be very close to the target point.
The procedure for computing the point in the intersection is summarized in Algorithm~\ref{POCS Algorithm}. $t_0$ is initialized to a random vector (representing the estimated location for the target point $q$) and is iteratively updated. ${\mathcal {P}}_{S_i}(t_{k})$ denotes the projection of point $t_k$ on sphere $S_i$. We generate the sequence of projections:
$$
{\displaystyle t_{k+1}={\mathcal {P}}_{S_{N}}\left({\mathcal {P}}_{S_{N-1}}(...{\mathcal {P}}_{S_1}(t_{k}))\right)},
$$

\begin{algorithm}[ht]
	\begin{algorithmic}[1]
		\caption{POCS Algorithm}
		\label{POCS Algorithm}
		\STATE{Initialize the maximum number of iteration $I_{\text{max}}$}
		\STATE{$t_0 = $ rand(1, D)} //D-dimensional random vector
		\STATE{$counter=0$}
		\REPEAT
		\FOR{$j$ = ${1\ to\ D+1}$}
		\STATE{$t_j = P_{S_j}(t_{j-1})$ //{P is projection\\ \ \ \ \ \ \ \ \ \ \ \ \ \ \ \ \ \ \ \ //of $t_{j-1}$ on $S_j$}}
		\ENDFOR
		\STATE{$counter++$}
		\UNTIL{Convergence == true or $counter==I_{\text{max}}$}
	\end{algorithmic}
\end{algorithm}

\section{Prior Art}
\label{sec:related}
PP-NNS is a heavily studied problem. However, existing solutions are limited with respect to at least one of the three requirements outlined in \sect{sec:intro}.
In addition to PP-NNS approaches discussed in \sect{sec:intro}, we briefly discuss most relevant prior works.
Several PP-NNS solutions are built upon the principals of cryptographically secure computation with the ability to compute on encrypted data~\cite{elmehdwi2014secure,chen2019sanns,songhori2015compacting,riazi2019mpcircuits,riazi2018chameleon,cramer2009multiparty,ppface,sadeghi2009efficient,evans2011efficient,larsen2019lower}.
The security of this approach, like cryptographic tools, is based on the hardness of certain problems in number theory (e.g. factorization of large numbers). Since every single bit in the computation is encrypted, distance calculations are computationally demanding and slow.

A line of work focuses on computing the distances using Paillier AHE scheme~\cite{paillier} and reporting the {\it nearest neighbor} to the query~\cite{ppface,sadeghi2009efficient,evans2011efficient}.
In~\cite{demmler2015aby}, authors leverage additive secret sharing to compute distances to find the closest point. 
The work~\cite{fingercode} enables a client to receive all similar entries (more than a pre-specified threshold).
In~\cite{zubertowards}, an approach based on TFHE~\cite{TFHE} is proposed to compute {\it argmin}. 
However, all of these solutions are several orders of magnitude slower than our proposed scheme and do not support an untrusted server.

In~\cite{chen2019sanns}, authors study the secure nearest neighbor search problem in which a client wants to query a database held by a server. The security model guarantees that no information about the query as well as the result is leaked to the server and clients learns nothing beyond the result of the search. However, it is assumed that the server holds the plaintext database which is in contrast to our security model where the server is not trusted. Authors use Homomorphic Encryption for distance computation and Garbled Circuits to identify minimum distances. 
In~\cite{riazi2019mpcircuits}, a solution for K nearest neighbors search is introduced where the database is {\it distributed} among many parties and a client wishes to query the aggregation of the database without revealing her input. The solution is based on the BMR protocol~\cite{beaver1990round}.

Another popular approach is to use information-theoretic secure multi-party computations, which guarantees that even with unlimited computational power no adversary can compromise the data. This method is based on secret-shared information to perform the secure computation and requires three or more servers. Securely computing pairwise distances needs ``comparison'' which is computationally intensive using secret-sharing alone and needs additional cryptographic blocks which limit the overall scalability~\cite{kilian1988founding,cramer2009multiparty}. These algorithms work by first computing all possible distances securely, before they find the near-neighbors based on minimum distance values. Irrespective of the underlying technique, calculating all distance pairs incurs $O(N)$ complexity.
Thus, the sub-linear time requirement cannot be satisfied by this class of techniques, rendering it unscalable to modern massive datasets.

There has been successful advances in the area of Differential Privacy (DP)~\cite{dwork2006differential,hardt2012simple,bourgeatlocal}. However, their security model and use cases of DP is different than ours. DP usually assumes a trusted server and aims to bound the information leakage when answering each query. In a very high level, a certain noise is added to the data stored on the database such that the statistical information of the database is preserved but an attacker cannot infer significant information about single entry in the database.

Order-Preserving Encryption (OPE)~\cite{boldyreva2009order,popa2013ideal} allows to carry out the comparison on encrypted version of data instead of the raw version. Wang et al.~\cite{wang2016practical} have proposed a solution based on OPE and R-tree for faster than linear PP-NNS. However, Naveed et al.~\cite{naveed2015inference} introduced several attacks that can recover original users' data from an encrypted database that are based on OPE or Deterministic Encryption (DTE). They have illustrated that the encrypted databases based on OPE or DTE are insecure.
Searchable encryption~\cite{song2000practical,curtmola2011searchable,kamara2012dynamic,kamara2013parallel,pappas2014blind,cash2014dynamic,prisearch2017} allows a user to store the encrypted data on the cloud server while being able to perform secure search. However, these solutions are limited to {\it exact} keyword search and are not compatible with NNS algorithms.

LSH is the algorithm of choice for sub-linear near-neighbor search in high dimensions~\cite{Proc:Indyk_STOC98}.
LSH techniques rely on randomized binary embeddings (or representations)~\cite{rane2010privacy,indyk2006polylogarithmic,aghasaryan2013use,petros,aumuller2018distance,aumuller2019fair}. These embeddings act as a probabilistic encryption which does not reveal direct information about the original attributes~\cite{indyk2006polylogarithmic, petros}. 
Due to the celebrated Jonson-Lindenstrauss~\cite{johnson1984extensions} or LSH property, it is possible to compare the generated embedding for a potential match.

\section{Conclusion}
This paper addresses the important problem of privacy-preserving near-neighbor search for multiple data owners while the query time is sub-linear in the number of clients. We show that the generic method of Locally Sensitive Hashing (LSH) for sub-linear query search is vulnerable to the triangulation attack. To secure LSH, a novel transformation is suggested based on the secure probabilistic embedding over LSH family. We theoretically demonstrate that our transformation preserves the near-neighbor embedding of LSH while it makes distance estimation mathematically impossible for non-neighbor points. By combining our transformation with Yao's Garbled Circuit protocol, we devise the first practical privacy-preserving near-neighbor algorithm, called Secure Locality Sensitive Indexing (SLSI) that is scalable to the massive datasets without relying on trusted servers. The paper provides substantial empirical evidence on real data from medical records of patients to online dating profiles to support its theoretical claims.

%

\balance

\bibliographystyle{ACM-Reference-Format}
\bibliography{slsh}

\setcounter{section}{0}
\section*{Appendices}
\section{Security of Black-Box Hash Computation}\label{sec:sbhc}
As we discussed in Section~\ref{sec:secure_hash}, the security of our black-box hash computation scheme is provided in Proposition~\ref{th:approachSecure}.

\begin{prop}\label{th:approachSecure}
	The proposed black-box hash computation scheme is secure in the honest-but-curious adversary model (standard security model in the literature) as long as two servers do not collude.
\end{prop}

{\bf Proof:}
The security proof of our scheme has two parts: (i) we need to prove that each server cannot infer {\it any} information about the actual client's input. (ii) {\it No} information about the final random seeds (that are used to compute hash value) are revealed neither to the servers nor the clients.

First, (i) is true because server \#1 gets the random value $v$ which is totally independent of $x$ and is generated randomly. Server \#2 receives $x \oplus v$ which is identical to the definition of one-time pad encryption and is proven to be secure (server \#2 doesn't have the encryption pad ($v$) and only receives $x \oplus v$). Please note that XORing the two secret shared values that are held by two servers ($v$ and $x \oplus v$), yields the client's input ($v \oplus (x \oplus v)=x$). But by the assumption of the non-colluding servers, this can never happen.

Second, (ii) is true because the GC protocol is a secure function evaluation protocol and by definition, at the end of the protocol none of the parties has any information about the other party's input. The inputs of server \#1 to the GC protocol consist of her random seeds and the pad $v$ while server \#2 inputs her random seeds along with $x \oplus v$. All other computations are done inside the GC protocol and are therefore secure. The final random seeds are the combination of the two random seeds from both servers inside the GC protocol and as a result, no one has access to its value.

\section{Circuit for Secure MinHash}\label{ssec:circuit_minhash}
As we discussed earlier, we need to design a Boolean circuit of the function that we want to evaluate securely. Here, we describe the circuit for computing secure MinHash. Figure~\ref{fig:circuit_minhash} shows the block diagram of the circuit. The circuit is designed based on the definition of secure MinHash. Our goal is to minimize the number of non-XOR gates due to the free-XOR technique~\cite{kolesnikov2008improved}. This technique makes the use of XOR gates almost free and hence the dominant cost metric is the number of non-XOR gates.

\noindent{\bf Inputs.} Each server puts her input as described bellow:
\begin{itemize}
	\item Server \#1: Binary vector $v$ of length $D$ that has been received from the client. Randomly generated $Seed_i\ \#1$ and $r_i$ where $i=1,\ 2,\ ...,\ k$.
	\item Server \#2: Binary vector $x \oplus v$ of length $D$ that has been received from the client. Randomly generated $Seed_i\ \#2$ and $r'_i$ where $i=1,\ 2,\ ...,\ k$.
\end{itemize}

\noindent{\bf Output.} 1-bit secure MinHash.
We now explain how the circuit works. First, $v$ and $x \oplus v$ are XORed to produce the real client's input ($x$) inside the circuit. Then two permutation seeds from two servers are used to permute the client's input. The final permutation is not revealed to either server since it is the combination of both permutations. The first seed together with $x$ are given to the first Waksman shuffling network~\cite{waksman1968permutation}. We have used this network to efficiently compute the random permutation with minimum number of non-XOR gates. The Waksman shuffling network is based on 2 input swapping blocks with 1-bit selection signal. To get the minimum number of non-XOR gates in the circuit, we have designed this block to have only one AND gate. Figure~\ref{fig:waksman_swap} shows the circuit for the swapping block. The swap signal is $s$, two inputs to the swap block are $a$ and $b$, and two outputs are $a'$ and $b'$. When $s=0$, $a'=a$ and $b'=b$ and when $s=1$, $a'=b$ and $b'=a$.
The number of swap blocks that are needed to implement the network is $F(N) = N\  log_2 N - N + 1$ where $N$ is the number of elements that are going to be shuffled (permuted)~\cite{waksman1968permutation} which in our case $N=D$.

\begin{figure}[ht]
	\centering
	\includegraphics[width=\columnwidth]{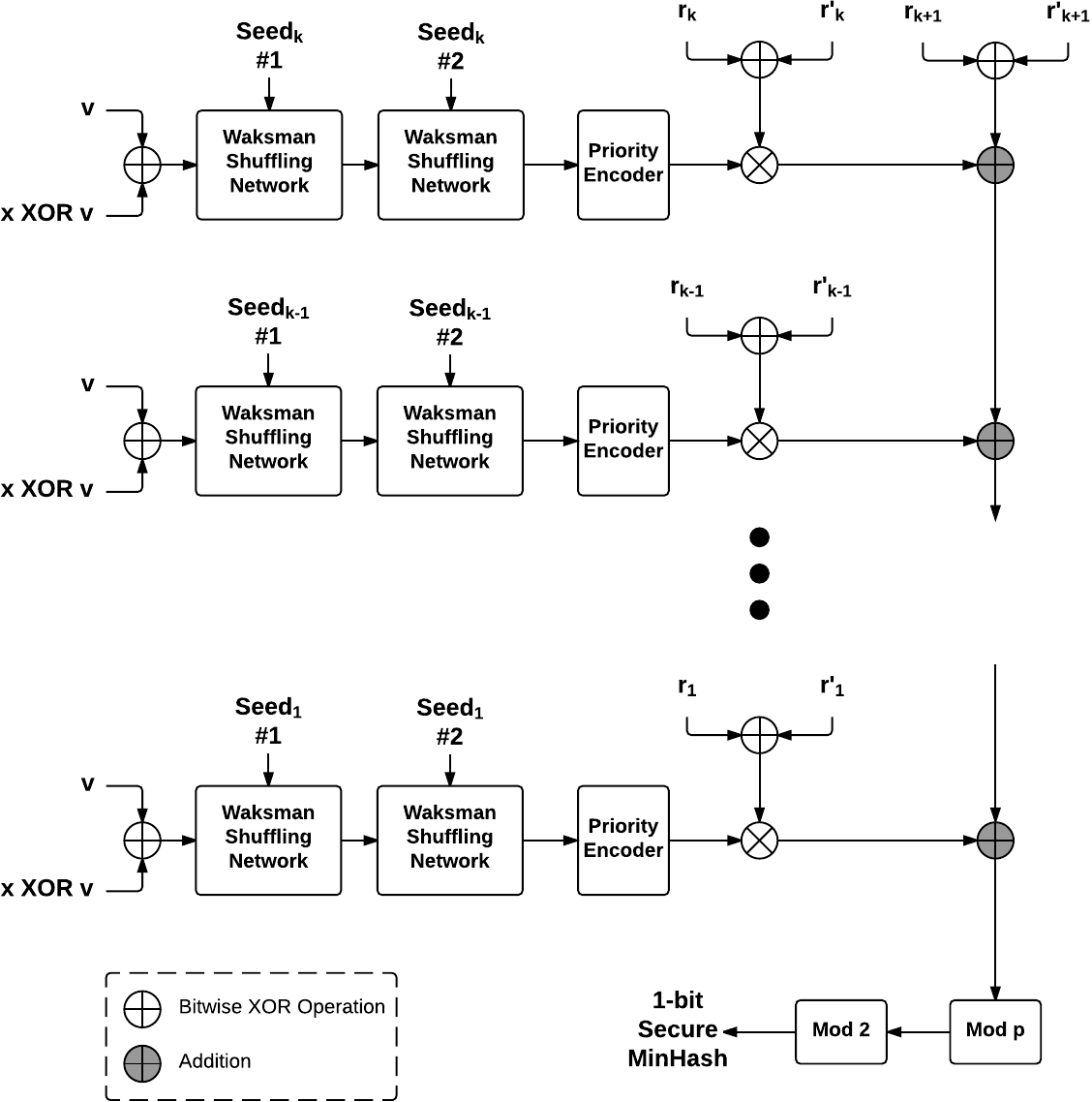}
	\caption{Boolean circuit for secure MinHash computation.
	}
	\label{fig:circuit_minhash}
\end{figure}

As the next step, we input the permuted client's input to the second Waksman network with the permutation seeds from the second server. The final permuted binary vector is then given to the priority encoder to find the index of the first non-zero element. Please note that from this point forward, the bit-length of wires are changed to $\log_2 D$ (because the index of the first non-zero element of $D$-bit binary vector needs to be described with $\log_2 D$ bits). At this point, we have computed a regular MinHash of client's input. Now, by the definition of secure MinHash, we need to multiply the hash with a random coefficient. Again, this coefficient is computed as the XOR of two random coefficients from two servers.

\begin{figure}[ht]
	\centering
	\includegraphics[width=0.5\columnwidth]{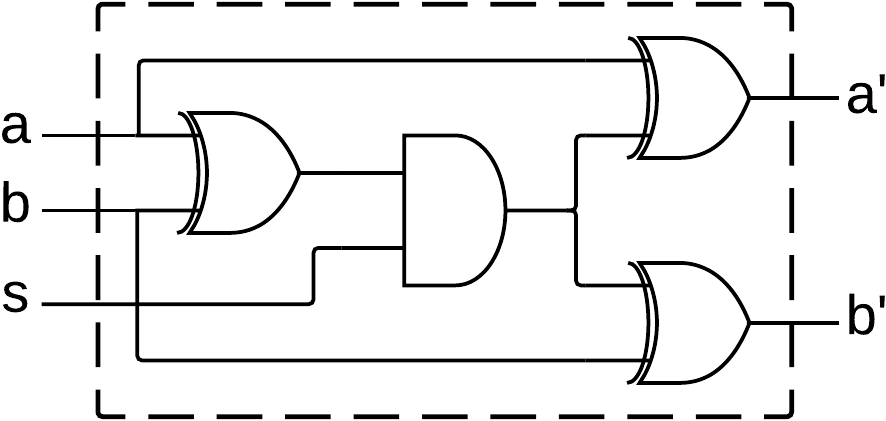}
	\caption{Boolean circuit used for swap block inside the Waksman shuffling network~\cite{waksman1968permutation}.
	}
	\label{fig:waksman_swap} 
\end{figure}

The same procedure is repeated $k$ times and the result of each step is added to the previous ones. In order to compute $mod\ p$ operation, we utilize a loop of subtraction to find the residue. The realization of $mod\ 2$ operation in a Boolean circuit can be performed by only outputting the Least Significant Bit (LSB) of the previous step. In the end, we have 1-bit secure MinHash. We need to run this circuit $l$ times to get $l$-bit secure MinHash.

\begin{table}[ht]
	\centering \normalsize
	\caption{Number of non-XOR gates in each block of the circuit.}
	\label{tab:circuit_cost}
	\begin{tabular}{l|c}
		Block Name & Number of non-XOR gates \\[0.1cm] \hline \hline \\[-0.3cm]
		Waksman Network & $D\times log_2 D - D + 1$ \\[0.1cm]
		Priority Encoder & $2\times D + (\frac{D}{2} -1)\times log_2 D$ \\[0.1cm]
		Multiplication & $(log_2 D)^2$ \\[0.1cm]
		Addition & $log_2 D$ \\[0.1cm]
		Mod p & $C_{mod}\times log_2 D$ $\dagger$\\[0.1cm]
		Mod 2 & 0 \\[0.1cm]
		Bitwise XOR & 0 \\[0.1cm]
	\end{tabular}
	\\
	\footnotesize{{$\dagger$} $C_{mod}$ is a constant that depends on the fixed prime $p$.}
\end{table}

{\bf Concrete Circuit Cost.}
We analyze the number of non-XOR gates in the circuit concretely and give a mathematical cost function based on our aforementioned parameters. Table~\ref{tab:circuit_cost} summarizes the costs.
All of the operations in Table~\ref{tab:circuit_cost} have to be performed $k$ times except for $mod\ p$ (one time) and Waksman network ($2\times k$ times). Therefore, the total cost (number of non-XOR gates) is given by the following formula:

\begin{align}\notag
&\#of\_non\_XOR\_gates =\\
k\times & ( \frac{5}{2}D\times log_2 D + (log_2 D)^2 + 2) + C_{mod}\times log_2 D
\end{align}

Please note that we need to run the circuit for $l$ times. Therefore, the total number of non-XOR gates in the GC protocol is $l$ times the total number of non-XOR gates in depicted circuit. Substituting $l=32$, $k=8$, and $D=1024$ ($2^{10}$), the total number of non-XOR gates is 6.58M. Using state-of-the-art GC frameworks such as TinyGarble~\cite{songhori2015tinygarble}, the end-to-end secure hash computation takes almost 1 second to finish.

\section{Circuit for Secure SimHash}\label{ssec:circuit_simhash}
In this section, we provide the Boolean circuit description of SimHash black-box computation. \fig{fig:circuit_simhash} shows the architecture of the circuit. Similar to the circuit of MinHash, server \#1 inputs vector $v$ of length $D$ along with her random seeds (Seed\_i \#1s and $r_i$s). Also, server \#2 inputs vector $x\oplus v$ and her random seeds (Seed\_i \#2s and $r'_i$s) to the circuit. Please note that each element of vector $v$ can be more than 1-bit as opposed to MinHash. Here, we implement our circuit for 32-bit fixed-point signed numbers. In contrast to MinHash circuit, the final random seeds in SimHash ($w_i$ vectors) are created by simply XORing the two random seeds from two servers. Since none of the servers has access to the other share, final $w_i$ vectors are not known to anyone. After calculating the vector dot product of $w_i$s and $x$, we extract the sign bit. Up to this point, we have computed $k$ different regular SimHashes. As described in \sect{sec:secSim}, we need to feed $k$ regular SimHash values to the universal hash function (right half of the circuit). The coefficients in the universal hash function are also the XOR of two servers' shares. 1-bit Secure SimHash is computed after $mod\ p$ and $mod\ 2$ modules.

\begin{figure}[ht]
	\centering
	\includegraphics[width=\columnwidth]{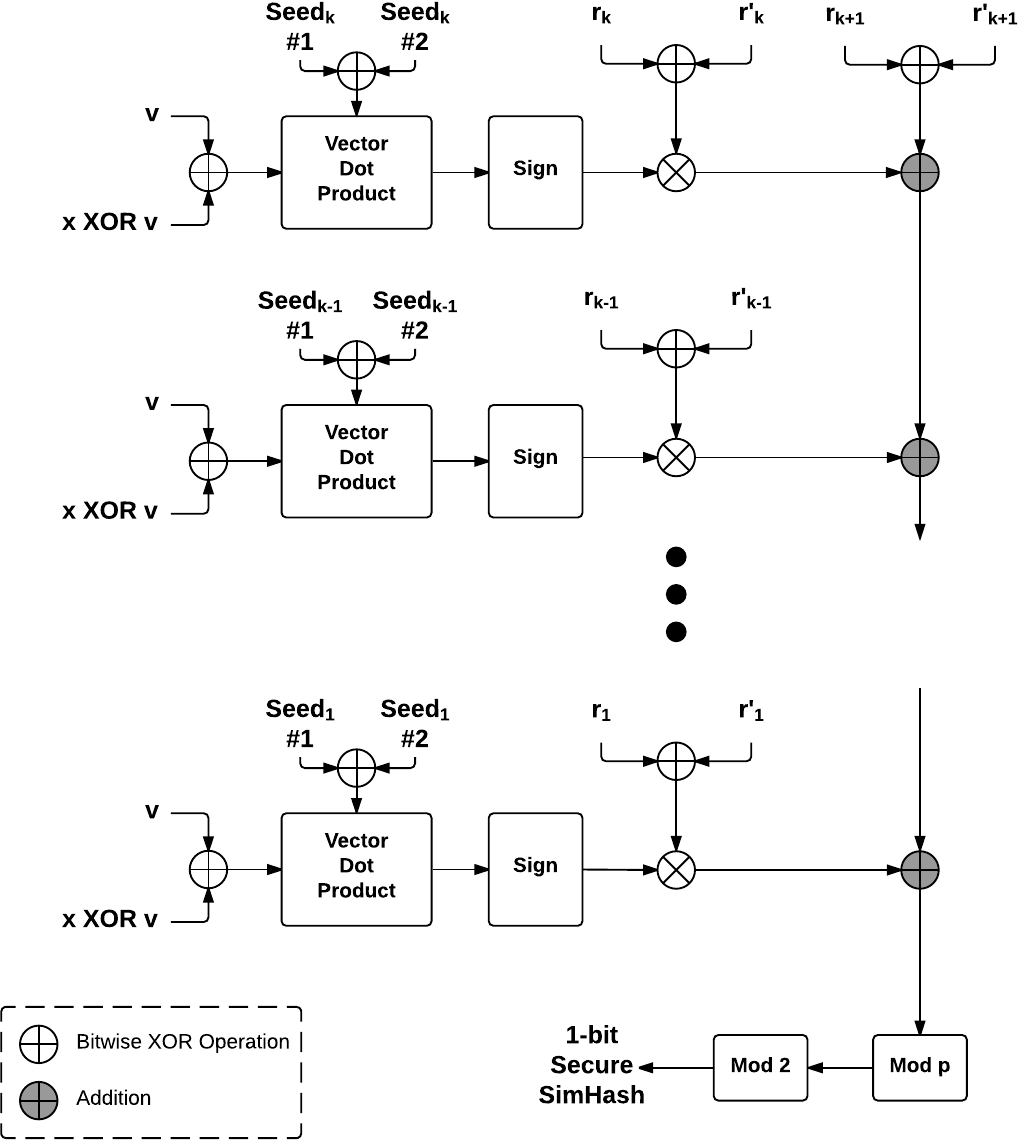}
	\caption{Boolean circuit for secure SimHash computation.
	}
	\label{fig:circuit_simhash}
\end{figure}

As mentioned in \sect{sec:secSim}, SimHash can be efficiently computed by randomly choosing values of $w_i$ between +1 and -1 with 0.5 probability~\cite{Book:Rajaraman_11}. Therefore, the ``Vector Dot Product'' of \fig{fig:circuit_simhash} can efficiently be implemented as $D-1$ 32-bit Addition modules. The ``Sign'' modules only output the sign bit and therefore does not require any Boolean gate to be implemented. In addition, since the output of the ``Sign'' module is only 1-bit, multiplying it with $r_i\oplus r'_i$ translates to changing the sign bit of $r_i\oplus r'_i$. These optimizations significantly reduce the number of gates in the circuit. Similar to secure MinHash, we need to evaluate the circuit $l$ times in order to create a $l$-bit secure SimHash.

\begin{table}[ht]
	\centering \normalsize
	\caption{Number of non-XOR gates in each block of the circuit.}
	\label{tab:circuit_cost_simhash}
	\begin{tabular}{l|c}
		Block Name & Number of non-XOR gates \\[0.1cm] \hline \hline \\[-0.3cm]
		Vector Dot Product & $32\times (D-1)$ \\[0.1cm]
		Multiplication & 0$\dagger$ \\[0.1cm]
		Addition & 32 \\[0.1cm]
		Mod p & $C_{mod}\times 32$ $\ddagger$\\[0.1cm]
		Mod 2 & 0 \\[0.1cm]
		Sign & 0 \\[0.1cm]
		Bitwise XOR & 0 \\[0.1cm]
	\end{tabular}
	\\
	\footnotesize{{$\dagger$} Because one operand is only a sign bit.}\\
	\footnotesize{{$\ddagger$} $C_{mod}$ is a constant that depends on the fixed prime $p$.}
\end{table}

{\bf Concrete Circuit Cost.}
We analyze the number of non-XOR gates in the circuit concretely and give a mathematical cost function based on our aforementioned parameters. Table~\ref{tab:circuit_cost_simhash} summarizes the costs.
All of the operations in Table~\ref{tab:circuit_cost_simhash} have to be performed $k$ times except for $mod\ p$ (one time). Therefore, the total cost (number of non-XOR gates) is given by the following formula:
\begin{align}\notag
&\#of\_non\_XOR\_gates =\\
k\times & ( 32\times D) + C_{mod}\times log_2 D
\end{align}

We run the GC protocol utilizing this circuit in the TinyGarble~\cite{songhori2015tinygarble} platform and the total execution time was 0.415 second for $k=12$, $l=32$, and $D=40$ (IWPC dataset parameters). Since each client only needs to perform this task once and independently of others, black-box hash computation only adds 0.415 second to the execution time of our end-to-end protocol.

\section{NNS Execution in GC}
\label{sec:newexp}
Secure K-Near-Neighbor (KNN) search has previously been studied by~\cite{songhori2015compacting}. They consider Hamming distance as their measure of similarity and utilize the GC protocol. In \sect{ssec:cwsfe}, we compared the performance of our scheme with theirs. However, we modify their approach by replacing the Hamming distance block with the {\it Cosine} similarity block. Also, in order to provide a fair comparison, we have changed the KNN search with the threshold-based NNS (same as ours). That is, instead of outputting the K nearest neighbors, we simply output whether two input data are more similar than a predefined threshold. Therefore, we have implemented a Boolean circuit that compares two input data and outputs the Boolean value one if their cosine similarity is more than a threshold. To find all near neighbors of a given query, we have to run the circuit for all the data on the server. Each time, we input the query together with one of the data in the database and after the GC protocol execution, we announce whether they are similar or not.

The number of non-XOR gates in the circuit for $D=40$ is $125,754$, where each value is represented as a 32-bit signed fixed-point number. In order to perform NNS on the database size of $N=3$ billion, $3.75\times 10^{14}$ non-XOR gates should be processed. Utilizing state-of-the-art GC-based framework~\cite{songhori2015tinygarble}, this task requires $1.5\times10^8$ second processing time and $1.2\times 10^7$ GBytes of communication.

\section{Compressed Sensing Lower Bounds}
\label{sec:CS}
Our protocol requires an assumption that multiple parties involved in the generating of $S(.)$ do not collude. 
However, even if parties collude and function $S(.)$ is compromised, it is significantly hard to invert $S(x)$. Revealing the exact mechanism of $S(.)$ poses a threat of possible inversion of the function $S(x)$ to obtain $x$ using Compressed Sensing. Note that there is a loss of information from $x$ to $S(x)$, due to heavy quantization and mod operations. Following are the two main reasons why such an inversion is hard even with the complete knowledge of $S(.)$:

\begin{enumerate}
\item {\bf Compressed Sensing Inference is Similar to Triangulation:} 
The algorithms for compressed sensing are only known for LSH-style signed measurements. Compressed sensing from our proposed secure LSH is not known, and we suspect it might be significantly difficult. Current compressed sensing algorithms work by iteratively reducing the possible space of the $x$, given $S(.)$, by introducing constraints about some known point~\cite{blumensath2009iterative}. This procedure is similar to triangulation attack. Thus, if triangulation attack does not work, it is unlikely that there will be an efficient compressed sensing algorithm. Compressed sensing attack for \sys{} is a topic for future work which could be of independent interest.

\item {\bf Compressed Sensing Lower Bounds:} Compressed sensing requires at
least $O(s\log{D})$ measurements for reasonable accuracy~\cite{candes2008introduction}, where $s$ is the number of non-zeros in the data vector and $D$ is the dimensionality of the data. Also, the big-O has large hidden constant. Thus, for high-dimensional vectors with significant non-zeros, the number of bits in $S(.)$ is smaller than $O(s\log{D})$, making it automatically secure against any possible attack since compressed sensing lower bounds are generic to any linear measurements~\cite{candes2008introduction}. Note that we only need few (constant number of) bit measurements to identify neighbors/non-neighbors in the Hamming distance. The number of bits required is independent of dimension and only depends on the similarity level.
\end{enumerate}

\end{document}

%% file: macros.tex
\ifsubmission
  \newcommand{\TODO}[1]{}
  
\else 
  \newcommand{\TODO}[1]{\textcolor{red}{TODO: #1}}
  
\fi

\newcommand{\remove}[1]{}

\newcommand{\sect}[1]{Section~\ref{#1}}
\newcommand{\fig}[1]{Figure~\ref{#1}}

\newcommand{\sys}[1]{SLSI}